\documentclass[]{IEEEtran}
\usepackage{fix-cm}
\usepackage{amsmath,amsfonts,amssymb,amsthm}
\usepackage{booktabs}
\usepackage{textcomp}
\usepackage{url}
\usepackage{graphicx}
\usepackage{cite}

\usepackage{ieeenotice}
\ieeenoticesetup{
  status  = accepted,   
  venue   = {IEEE Open Journal of Communications Society},
  doi     = {10.1109/OJCOMS.2026.3757746},
  license = cc-by,       
}

\newcommand{\eqdef}{\triangleq}
\newcommand{\Zset}{\mathbb{Z}}
\newcommand{\Cset}{\mathbb{C}}
\newcommand{\Tu}{T_{\text{\rm u}}}
\newcommand{\Ts}{T_{\text{\rm s}}}
\newcommand{\Tg}{T_{\text{\rm cp}}}
\newcommand{\Tgamma}{T_{\gamma}}
\newcommand{\Tr}{T_{\text{\rm r}}}

\newcommand{\fc}{f_{\text{\rm c}}}
\newcommand{\fs}{f_{\text{\rm s}}}
\renewcommand{\Re}{\operatorname{Re}}
\newcommand{\sinc}{\text{\rm sinc}}
\newcommand{\rect}{\text{\rm rect}}
\newcommand{\trasp}{\text{T}}
\newcommand{\herm}{\text{H}}
\newcommand{\eb}{\mathbf{e}}
\newcommand{\eh}{\eb_h}
\newcommand{\gammab}{\boldsymbol{\gamma}}
\newcommand{\Gammab}{\boldsymbol{\Gamma}}
\newcommand{\Hb}{\mathbf{H}}
\newcommand{\ts}{t_\text{s}}
\newcommand{\deltas}{\delta_{\text{s}}}
\newcounter{mytempeqncnt}

\newcounter{assumption}
\renewcommand{\theassumption}{a\arabic{assumption}}
\newcommand{\assumption}[2][]{%
    \refstepcounter{assumption}%
    \ifx\relax#1\relax\else\label{#1}\fi
    \textbf{(\theassumption)}~#2%
}

\newtheorem{theorem}{Theorem}
\theoremstyle{remark}
\newtheorem{remark}{Remark}

\begin{document}

\title{Physically Consistent Modeling of Dispersive
Time-Modulated Reconfigurable Intelligent Surfaces
for Wideband OFDM}

\author{Ivan Iudice,
Donatella Darsena,
Giacinto Gelli,
and Vincenzo Galdi
\thanks{I. Iudice is with the
Security and Cybersecurity in Complex Systems and Air Traffic Management Unit,
Italian Aerospace Research Centre (CIRA), Capua (CE) 81043, Italy.}
\thanks{D. Darsena and G. Gelli are with the
Department of Electrical Engineering and Information Technology,
University of Napoli Federico II, Napoli 80125, Italy.}
\thanks{V. Galdi is with the
Fields \& Waves Lab, Department of Engineering,
University of Sannio, Benevento 82100, Italy.}
\thanks{The work of I. Iudice was supported by the national project
SUBAVSTM (Supporto nel Volo Sub-Orbitale, Aviolancio e Space Traffic Management),
with grant PRORA (Programma Nazionale di Ricerche Aerospaziali) DM662.}}

\maketitle

\begin{abstract}
The elements of a reconfigurable intelligent surface (RIS)
are commonly modeled either as frequency-selective
time-invariant reflectors or as 
instantaneous time-varying complex-valued reflection coefficients.
However, in practical implementations, 
time-modulated metasurfaces simultaneously exhibit
frequency selectivity and periodic time
variation. This paper develops a physically consistent linear
periodically time-varying (LPTV) model that jointly captures
these effects and characterizes their impact on wideband
orthogonal frequency-division multiplexing (OFDM) communications.
Starting from a canonical equivalent-circuit representation of
a generic RIS element, we derive a single-resonance
model whose physically meaningful parameters
determine both the frequency-selective reflection coefficient
and the effective impulse-response duration. This circuit-based
description therefore characterizes explicitly
the finite memory introduced by the RIS element.
The periodically switched dispersive responses are subsequently
represented through harmonic transfer functions, leading to a
closed-form per-subcarrier OFDM input--output relation. The
resulting harmonic coupling 
is generally non-diagonal: each received
subcarrier collects contributions from multiple transmitted
subcarriers through the RIS harmonics, with each contribution
weighted by the element response evaluated at the corresponding
absolute input frequency.
We further derive a generalized cyclic-prefix (CP) condition
requiring the guard interval to accommodate the sum of the
propagation-channel delay spread and the RIS memory. Under this
condition, intersymbol interference and out-of-grid spectral
leakage are suppressed, while deterministic on-grid harmonic
coupling remains.
Full-wave simulations of a unit cell from the OpenRIS repository,
designed for a 5G NR channel, validate the proposed resonant model.
They also reveal appreciable in-band dispersion
even though the unit cell provides nearly ideal phase switching
between two discrete phase states.
Simulations over the 3rd Generation Partnership Project
(3GPP) tapped-delay-line channels confirm the generalized CP condition
and illustrate its relevance for high-quality-factor RIS elements.
\end{abstract}


\section{Introduction}
\IEEEPARstart{R}{econfigurable} intelligent 
surfaces (RISs) have emerged as a promising
technology for shaping the wireless 
propagation environment in beyond-5G
and future 6G networks~\cite{Wu2020,DiRenzo2020}. 
An RIS is typically composed of a large number of individually 
reconfigurable sub-wavelength 
elements or unit cells, whose collective response determines 
the overall transformation applied by the surface to the 
incident electromagnetic (EM) field.

Most signal processing-oriented 
models for RIS elements rely on one of two
idealized abstractions.
In the first,  
each element 
is modeled as a simple phase shifter 
with an instantaneous or frequency-flat response
\cite{Basar2019,Wu2019}.
Wideband extensions refine this abstraction by introducing 
a frequency-selective but time-invariant 
reflection coefficient~%
\cite{Cai2020,Abbas2025,Li2025,Katsanos2022,Katsanos2024},
also in the context of beyond-diagonal architectures~\cite{Sena2025,Katsanos2026};
for orthogonal frequency-division multiplexing (OFDM), 
an explicit per-element delay is introduced, which is 
constrained to lie within the OFDM cyclic prefix (CP) duration~\cite{An2021}.
A signal-processing framework accounting for 
filtering performed by each element was developed in \cite{Bjornson2022}.
However, under the narrowband-element assumption, 
the element impulse response reduces to a scaled 
and delayed Dirac delta pulse, 
thereby neglecting the unit cell memory. 
The second abstraction applies
to time-modulated  RISs, 
whose elements are periodically reconfigured 
to generate spectral harmonics~\cite{Zhang2018,Xiao2024},
for applications such as index modulation, beamforming, 
wireless transmission, and sensing.
Within this framework, each RIS element is 
modeled  as an instantaneous (i.e., frequency-flat)  
time-varying system \cite{Zhang2018,Dai2018}.

Physically feasible metasurfaces violate both idealizations. 
Their elements are inherently 
dispersive due to finite memory and, 
when periodically reconfigured, 
exhibit a periodically 
time-varying behavior. 
In the electromagnetic literature, these two phenomena 
have been jointly described at the field level using 
Floquet analysis with  
dispersive surface susceptibilities~\cite{tiu2021}.
A substantial body of work has addressed
physically consistent EM modeling of RISs,
based on multiport network theory
(impedance, admittance and scattering parameters),
polarizability descriptions, and equivalent circuits,
with extensions to mutual coupling and beyond-diagonal architectures.
Such models are typically developed under a time-invariant narrowband 
assumption,
and a recent survey of this literature \cite{Bidabadi2025}
explicitly identifies wideband operation
as one of the least explored directions within this framework.

Time variation is considered in \cite{Bjornson2022}
only as an uncontrolled effect due to terminal mobility,
described through linear time-varying (LTV) 
system theory and possibly compensated
by the surface itself, whereas the deliberate periodic reconfiguration
of the elements is listed there among the open directions
whose input--output description remains to be established.
However, to the best 
of our knowledge,
a signal processing-oriented 
input--output model that 
simultaneously captures both effects and 
enables the derivation of the corresponding 
per-subcarrier relation observed 
at a wideband OFDM receiver is still lacking.

In this paper, we bridge this 
gap by modeling each RIS element 
as a \textit{linear periodically time-varying} (LPTV) system 
characterized by  a two-dimensional reflection kernel 
and by deriving  the corresponding input--output relation when
the RIS is employed in a wideband OFDM system. 
The main contributions of this 
work are summarized as follows: 
(i) we characterize the finite-memory property 
of each RIS element using  
a simple physics-based circuit model;
(ii) we derive a Fourier-series 
representation of  the dispersive 
element response that  
preserves its frequency selectivity; 
(iii) we obtain a closed-form per-subcarrier 
input--output relation for OFDM signals 
reflected by a time-varying
RIS;
and (iv) we derive in Theorem~\ref{thm:guard} a generalized
CP condition, showing 
that the CP must accommodate
the combined delay spread introduced by the 
propagation channel and the RIS memory to
suppress intersymbol interference (ISI) 
and channel-induced 
intercarrier interference (ICI), leaving as sole
residual impairment the structured ICI
inherent to the
RIS periodic reconfiguration.

The remainder of this paper is organized 
as follows. 
Section~\ref{sec:lptv-model}
introduces the proposed LPTV 
model for dispersive time-varying RIS elements. 
Section~\ref{sec:ofdm-analysis} derives the 
corresponding input--output relation for 
wideband OFDM signals and establishes 
the generalized CP condition. 
Section~\ref{sec:simulations} presents numerical results validating 
the proposed analysis. 
Finally, Section~\ref{sec:concl} concludes the paper.

\textit{Notation:} Boldface letters denote vectors and matrices;
$(\cdot)^*$, $(\cdot)^\trasp$, and $(\cdot)^\herm$ denote conjugation,
transposition, and conjugate transposition, respectively;
$\Zset$ and $\Cset$ are the sets of integer and complex numbers,
respectively; $\Re\{\cdot\}$ denotes the real part and
$\mathbb{E}[\cdot]$ the statistical expectation;
$\delta(\cdot)$ is the Dirac delta and $\delta_k$ the Kronecker delta,
$\text{u}(\cdot)$ the unit step function,
and $\rect(t)$ the rectangular window
equal to one for $t\in[-1/2,1/2)$ and zero otherwise;
$\log(\cdot)$ is the natural logarithm;
finally, $\sinc(x)\eqdef\sin(\pi x)/(\pi x)$.

\section{LPTV Modeling of the RIS Elements}
\label{sec:lptv-model}

Let us consider a RIS comprising $Q$ elements, 
and denote with $x_{\text{bp},q}(t)$ 
and $y_{\text{bp},q}(t)$, respectively, 
the real-valued input and output signals
at the $q$-th RIS element,
both assumed to be \textit{bandpass} (bp) 
signals with carrier
frequency $\fc$.
Throughout the paper, the $e^{j2\pi f t}$
time-harmonic convention is adopted,
consistently with the Fourier transform definition
$X(f) = \mathcal{F}[x(t)] \eqdef \int_{-\infty}^{+\infty}
x(t)\, e^{-j2\pi f t}\, dt$.

Following the classical theory of LTV
systems~\cite{Zadeh1950,Zadeh1961}, the related
signal-processing formulations \cite{GelliVerde2002},
and the time--frequency representation 
introduced by Bello~\cite{Bello1963}\footnote{
See also \cite{Bjornson2022} for the use
of the same formalism in RIS-aided links under receiver mobility.}, 
the response of the $q$th element 
is modeled as 
\begin{equation}
y_{\text{bp},q}(t) = \int_{-\infty}^{+\infty}
\Gamma_q(f,t) \, X_{\text{bp},q}(f) \,
e^{j2\pi f t} \, df, 
\label{eq:tvfilter}
\end{equation}
where  $X_{\text{bp},q}(f) = \mathcal{F}[ x_{\text{bp},q}(t)]$.
The two-dimensional transfer function  
$\Gamma_q(f,t)$ jointly captures  
the frequency-selective response of the element
(through its dependence on $f$) 
and the intentional time variation 
(through its dependence on $t$)
induced by RIS reconfiguration.

Using the complex-envelope representation
$x_{\text{bp},q}(t) = \Re \left\{ x_q(t)
\, e^{j2\pi \fc t} \right\}$,
with $x_q(t)$ denoting the complex baseband signal,
whose spectrum $X_q(f)=\mathcal{F}[x_q(t)]$ is assumed to vanish
outside the interval $\lvert f \rvert \le B/2$,
with $B \ll \fc$ denoting the signal bandwidth,
the input bandpass spectrum in \eqref{eq:tvfilter}
can be expressed as
$
X_{\text{bp},q}(f)
=
\frac{1}{2}X_q(f-\fc)
+
\frac{1}{2}X_q^{*}(-f-\fc).
$
Exploiting the Hermitian symmetry
$\Gamma_q^{*}(f,t)=\Gamma_q(-f,t)$,
which follows from the assumption that
the impulse response of the
reflective element is real-valued
(see also discussion in Section \ref{sec:one-pole} and 
Appendix~\ref{sec:surf-response}),
the output complex baseband signal
$y_q(t)$, defined through
$y_{\text{bp},q}(t)
=
\Re\left\{
y_q(t) \, e^{j2\pi \fc t}
\right\}$,
can be expressed as
\begin{equation}
y_q(t)
=
\int_{-\infty}^{+\infty}
\Gamma_q(\nu+\fc,t) \, X_q(\nu) \, e^{j2\pi\nu t} \, d\nu.
\label{eq:bb}
\end{equation}

\subsection{Harmonic Representation}

Similarly to the space--time coding metasurface architecture proposed in \cite{Zhang2018},
let us assume that each RIS element is driven by a periodic 
control signal of period $T=K\Tr$, 
which is piecewise constant over $K$ 
slots of duration $\Tr$.
Under the \textit{quasi-static approximation} 
(see Remark~\ref{rem:quasi-static}) the element response within
each modulation slot can be approximated by the steady-state transfer
function associated with the corresponding control state.
Consequently, $\Gamma_q(f,t)$ is periodic in $t$
and can be expressed as
\begin{equation}
\Gamma_q(f,t) = \sum_{n\in\Zset}
\sum_{k=0}^{K-1} \Gamma_q^{[k]}(f) \,
\rect\left(\frac{t-k\Tr-nT-\Tr/2}{\Tr}\right),
\label{eq:piecewise}
\end{equation}
where $\Gamma_q^{[k]}(f)$ 
denotes the frequency-selective
steady-state reflection coefficient
associated with the control state applied 
during the $k$-th modulation slot.
%


\begin{remark}[Quasi-static approximation]
\label{rem:quasi-static}
The piecewise-constant representation in \eqref{eq:piecewise}
implicitly assumes that, within each control slot,
the RIS element operates in a quasi-static regime 
described by the LTI response $\Gamma_q^{[k]}(f)$.
This approximation is valid provided that the transient
generated by each switching event has essentially 
vanished before the subsequent state transition.
\end{remark}

\noindent
Further details on 
the steady-state condition 
and its relationship with the classical
\textit{adiabatic condition}~\cite{Zhang2018} are 
given in Remark~\ref{rem:steady-state} later on.
Since $\Gamma_q(f,t)$ is periodic in
$t$ with period $T$, it admits 
the Fourier-series expansion
\begin{equation}
\Gamma_q(f,t)=\sum_{h\in\Zset} b_q^{[h]}(f) \,
e^{j2\pi\frac{h}{T}t},
\label{eq:fs}
\end{equation}
where the Fourier coefficients are given by
\begin{equation}
b_q^{[h]}(f) = \tfrac{1}{K} e^{-j\pi\frac{h}{K}}\,
\sinc\left(\tfrac{h}{K}\right) \sum_{k=0}^{K-1}
\Gamma_q^{[k]}(f) \, e^{-j\frac{2\pi}{K}hk} \: .
\label{eq:fs-inv}
\end{equation}
For notational convenience, let us define the vectors
$\Gammab_q(f) \eqdef [\Gamma_q^{[0]}(f),
\Gamma_q^{[1]}(f), \ldots, \Gamma_q^{[K-1]}(f) ]^\trasp
\in \Cset^K$ 
and
$
\eh \eqdef \frac{1}{K} \, e^{j\pi\frac{h}{K}}
\, \sinc\left(\frac{h}{K}\right)
[ 1,\,e^{j\frac{2\pi}{K}h},
\dots,e^{j\frac{2\pi}{K}h(K-1)}]^\trasp
\in \Cset^{K}
$.
Then, the $h$-th harmonic coefficient 
can be compactly written as
$b_q^{[h]}(f)=\eh^\herm\Gammab_q(f)$,
and \eqref{eq:fs} becomes
\begin{equation}
\Gamma_q(f,t)=\sum_{h\in\Zset} \eh^\herm \,
\Gammab_q(f) \, e^{j2\pi\frac{h}{T}t}.
\label{eq:fs-compact}
\end{equation}
Note that $\Gammab_q(f)$ does not depend on $h$, since it
collects the frequency responses associated with the $K$
control states. The dependence on the harmonic index is
instead captured by $\eh$, whose inner product with
$\Gammab_q(f)$ yields the $h$-th Fourier coefficient.

Equation~\eqref{eq:fs-compact} constitutes 
the core of the proposed model.
The harmonic index $h$ identifies
the spectral components generated 
by the periodic RIS reconfiguration,
whereas the dependence of $\Gammab_q(f)$ on $f$
captures the frequency-selective behavior of the 
reflective element.
Moreover, the proposed formulation naturally encompasses  
existing models as special cases.
Specifically, if the slot-dependent 
reflection coefficients are assumed to be frequency flat, i.e., 
$\Gamma_q^{[k]}(f) = \Gamma_q^{[k]}$ for every $k$,
the model reduces to the conventional  frequency-flat time-modulated 
RIS representation~\cite{Zhang2018}.
Conversely, in the absence of temporal modulation, 
all control states share the same response
$\Gamma_q^{[k]}(f) = \bar\Gamma_q(f)$,
and all non-zero harmonics vanish, so that the model
reduces to  a frequency-selective
LTI one~\cite{Katsanos2022,Abbas2025,Katsanos2024}.


\subsection{Equivalent Circuit and One-Pole Model}
\label{sec:one-pole}

To relate the steady-state frequency 
response $\Gamma_q^{[k]}(f)$
to a physically realizable RIS element,
we resort to the canonical 
equivalent-circuit model proposed in 
\cite{Costa2021}.\footnote{
This loaded-surface description belongs to the family
of physically consistent EM models of RISs surveyed in \cite{Bidabadi2025}.
Like the full-wave characterization of Section~\ref{sec:simulations},
which uses Floquet boundary conditions, it describes a periodic surface of
identical cells, so that the mutual coupling among cells in the same state
is embedded in the effective sheet parameters; the case of neighboring cells
in different states is discussed in Remark~\ref{rem:bdris}.}
Specifically, letting 
$\omega \eqdef 2\pi f$,
a reflective element 
printed on a grounded dielectric substrate
can be described by the following 
loaded-surface impedance model:
\begin{equation}
Z_{\text{in}}^{\text{pol}}(\omega,\theta)
= Z_{d}^{\text{pol}}(\omega,\theta) 
\parallel 
Z_{\text{patch}}^{\text{pol}}(\omega,\theta) 
\parallel  
Z_{\text{var}}(\omega),
\label{eq:Zin}
\end{equation}
where $\theta$ denotes 
the incidence angle of the impinging
EM wave, 
the superscript 
$\text{pol}\in\{\text{TE},\text{TM}\}$
identifies its (transverse-electric and transverse-magnetic, respectively) polarization,
and $A\|B=AB/(A+B)$ denotes the parallel operator.
The three impedances in \eqref{eq:Zin}  
represent the grounded dielectric substrate, 
the metallic patch, 
and the tunable load 
(e.g., a varactor diode), respectively. 
Their detailed expressions are reported in 
Appendix \ref{sec:surf-response}.

The corresponding reflection coefficient is obtained  
using the polarization-dependent free-space wave impedance 
$Z_0^{\text{pol}}(\theta)$ as
\begin{equation}
\Gamma^{\text{pol}}(\omega,\theta)=
\frac{Z_{\text{in}}^{\text{pol}}(\omega,\theta)-Z_0^{\text{pol}}(\theta)}
{Z_{\text{in}}^{\text{pol}}(\omega,\theta)+Z_0^{\text{pol}}(\theta)},
\quad \text{pol} \in \{\text{TE},\text{TM}\},
\label{eq:Gmap}
\end{equation}
where $Z_0^{\text{TE}}(\theta) \eqdef\ \frac{\zeta_0}{\cos\theta}$
and $Z_0^{\text{TM}}(\theta) \eqdef \zeta_0\cos\theta$ \cite{Balanis2016},
with $\zeta_0 \approx 377 \, \Omega$ denoting 
the free-space intrinsic impedance.

Equation~\eqref{eq:Gmap} returns the steady-state reflection coefficient
of the cell for a given setting of the tunable load: it coincides with
$\Gamma_q^{[k]}(f)$ in \eqref{eq:piecewise}, evaluated at $\omega = 2\pi f$,
when the load is set to the $k$-th control state.
The kernel $\Gamma_q(f,t)$ in \eqref{eq:tvfilter} is built 
from these steady-state responses through \eqref{eq:piecewise}.

A well-designed RIS element 
behaves as a resonant circuit
whose resonance frequency is close 
to the operating frequency. 
Consequently, in the vicinity of $\fc$, 
the equivalent-circuit model can be reduced to a 
{single-resonance} 
representation 
under the following
assumptions \cite{Munk2000,Abeywickrama2020}:
\assumption[ass:norm]{normal incidence and symmetric cell};
\assumption[ass:thin-sub]{electrically thin substrate};
\assumption[ass:cap-load]{capacitively loaded sheet};
\assumption[ass:loss]{lumped losses};
\assumption[ass:narrow]{narrowband signals}.

Under these assumptions, 
the reflection coefficient of a 
single resonant element takes
the familiar Lorentzian pole--zero form
derived in Appendix \ref{sec:surf-response}:
\begin{equation}
\begin{split}
\Gamma_q^{[k]}(f) &
=-\,\frac{j2\pi\left(f-f_{0,q}^{[k]}\right)+\xi_{\text{i},q}^{[k]}
-\xi_{\text{r},q}^{[k]}}%
{j2\pi\left(f-f_{0,q}^{[k]}\right)+\xi_{\text{i},q}^{[k]}
+\xi_{\text{r},q}^{[k]}} \\
& =-1+\frac{2\xi_{\text{r},q}^{[k]}}%
{j2\pi\left(f-f_{0,q}^{[k]}\right)
+\xi_{q}^{[k]}} \: .
\end{split}
\label{eq:polezero}
\end{equation}
Here, with reference to the $q$-th RIS element and the 
$k$-th modulation slot, 
$f_{0,q}^{[k]}$ denotes the resonance frequency
and $\xi_{q}^{[k]} \eqdef \xi_{\text{r},q}^{[k]}
+ \xi_{\text{i},q}^{[k]}$ is the total decay rate, 
with $\xi_{\text{r},q}^{[k]}$ and $\xi_{\text{i},q}^{[k]}$ 
denoting the radiative (external) and intrinsic (internal) decay 
rates \cite{Fan2003}, respectively.
These parameters depend on the control state, 
selected, e.g., through the bias voltage 
applied to the varactor diode.

For notational simplicity, \eqref{eq:polezero} refers to a fixed
polarization, and the corresponding polarization index is
suppressed. 
Under normal incidence ($\theta = 0$) 
and for a symmetric 
unit cell, the TE and TM responses coincide. 
For oblique incidence ($\theta \neq 0$)
the parameters
$f_{0,q}^{[k]}$, $\xi_{\text{r},q}^{[k]}$, and
$\xi_{\text{i},q}^{[k]}$ become dependent on the incidence
angle and polarization and can be obtained by re-evaluating
\eqref{eq:Zin}--\eqref{eq:Gmap} 
without altering
either the single-resonance structure in~\eqref{eq:polezero} or
the proposed LPTV formulation. All full-wave results reported hereafter are at normal incidence; the oblique-incidence characterization is left to future work.

For a physically realizable RIS element with a real-valued
impulse response, the steady-state response $\Gamma_q^{[k]}(f)$ 
associated with
each control state must satisfy
Hermitian symmetry.
Since the switching functions 
in~\eqref{eq:piecewise} are real-valued, this implies the Hermitian symmetry
of $\Gamma_q(f,t)$. 
Equivalently, the
harmonic transfer functions must satisfy
$ b_q^{[h]}(-f)
=
\left[b_q^{[-h]}(f)\right]^*$.
The single-pole approximation in 
\eqref{eq:polezero}, however, retains
only the resonance centered at 
the positive frequency $f_{0,q}^{[k]}$ 
and therefore does not satisfy the
Hermitian symmetry.
Such a symmetry can be 
restored by including the mirror 
resonance at ($-f_{0,q}^{[k]}$),
leading to the Hermitian two-pole response 
reported in Appendix~\ref{sec:surf-response} [see \eqref{eq:Hcomplete}].
Nevertheless, over the signal band 
$|f-\fc| \le B/2$,
the contribution of the negative-frequency
resonance is nearly constant
and can be absorbed into the background reflection term. 
Consequently, the single-pole expression in 
\eqref{eq:polezero} approximates 
the physical (Hermitian) cell response
with relative error 
$\mathcal{O}(1/Q_\text{L})$, 
where $Q_{\text{\rm L}}$ is the loaded quality factor [see \eqref{eq:Q}],
as shown in
Appendix \ref{sec:surf-response}.
Accordingly, all subsequent developments 
evaluate $\Gamma_q^{[k]}(f)$
only at positive in-band frequencies,
where the single-pole approximation is accurate
to this order.


\subsection{Impulse Response and Finite Memory}
\label{sec:memory}

Let $\gamma^{[k]}_q(\tau) \eqdef
\int_{-\infty}^{+\infty} \Gamma_q^{[k]}(\nu+\fc) \,
e^{j2\pi\nu \tau} \, d\nu$ denote the equivalent
baseband steady-state 
impulse response associated 
with the $k$-th modulation slot
of the $q$-th RIS element.
By straightforward inverse Fourier transformation 
of \eqref{eq:polezero}, one gets
\begin{equation}
{\gamma}_q^{[k]}(\tau) =
-\delta(\tau) + 2\,\xi_{\text{r},q}^{[k]} \, e^{-\xi_q^{[k]}\tau}
\, e^{j2\pi\Delta_q^{[k]}\tau} \text{u}(\tau),
\label{eq:impulse-response}
\end{equation}
where 
$\Delta_q^{[k]} \eqdef f_{0,q}^{[k]} - \fc$ denotes
the operating detuning.

The exponential decay in \eqref{eq:impulse-response}
suggests defining an effective 
memory duration for each RIS element,
such that
${\gamma}^{[k]}_q(t) \approx 0$, $\forall\,t\not\in[0, T_{\gamma,q})$
and $\forall k \in \{0,1,\ldots, K-1\}$.
For a prescribed 
amplitude threshold 
$0<\chi<1$, 
the exponential term of the 
impulse response falls below 
a fraction $\chi$ of its initial value 
after a time
\begin{equation}
T_{\gamma,q} \eqdef \max_{0 \le k < K} 
\frac{\log(1/\chi)}{\xi_q^{[k]}} \: .
\label{eq:Tgamma-q}
\end{equation}
The effective memory in 
\eqref{eq:Tgamma-q} is governed solely by 
the total decay rate $\xi_q^{[k]}$ and
is therefore independent of the 
operating detuning $\Delta_q^{[k]}$.
In the remainder of the paper, each element response is
implicitly truncated to $[0,T_{\gamma,q})$. The neglected
resonant tail has an amplitude no greater than a fraction
$\chi$ of its initial value; for example, $\chi=10^{-3}$
corresponds to a $-60$-dB amplitude threshold.

To impose a common settling condition across the entire RIS,
we define its effective memory as
\begin{equation}
\Tgamma
\eqdef
\max_{1\leq q\leq Q}T_{\gamma,q}
=
\max_{\substack{1\leq q\leq Q\\0\leq k<K}}
\frac{\log(1/\chi)}{\xi_q^{[k]}} .
\label{eq:Tgamma}
\end{equation}
This worst-case definition does not require the RIS elements
to have identical responses. Rather, it ensures that the
slowest-decaying element, over all control states, is included
in both the settling-time and CP requirements.

\begin{remark}[Steady-state condition, memory, and adia\-ba\-ti\-city]
\label{rem:steady-state}
The effective memory $\Tgamma$ 
provides a direct physical interpretation
of the validity condition 
for the piecewise-constant
approximation in~\eqref{eq:piecewise}. In particular, the
transient generated by each switching event must essentially
decay before the subsequent state transition. A sufficient
condition is therefore
\begin{equation}
\Tr\gg\Tgamma .
\label{eq:settle}
\end{equation}
Since
$\Tgamma = \mathcal{O}\left( Q_{\text{\rm L}} \, \fc^{-1} \right)$
and $Q_{\text{\rm L}} \gg 1$,
it follows that $\Tgamma \gg \fc^{-1}$.
Consequently, \eqref{eq:settle} implies the classical adiabatic
condition for time-modulated metasurfaces~\cite{Zhang2018}:
\begin{equation}
\Tr^{-1} \ll \fc .
\label{eq:adiab}
\end{equation}
\end{remark}

\noindent 
Remark~\ref{rem:steady-state} shows that
the carrier-rate condition~\eqref{eq:adiab} is not an
independent assumption, but a consequence of the more
fundamental requirement in \eqref{eq:settle} 
that the resonator reaches its steady state 
before the next switching event.
The effective memory $\Tgamma$
quantifies the temporal spreading introduced by the RIS,
playing a role analogous to the delay spread of a
wireless propagation channel.
As shown in the next section,
this additional memory contributes to the overall
delay spread experienced by the OFDM waveform and,
consequently, enters the CP design condition.

Finally, it is worth noting that, 
under \eqref{eq:settle} in Remark~\ref{rem:steady-state},
the piecewise-constant model \eqref{eq:piecewise}
is accurate for quantities averaged over a control slot, such as the
demodulated OFDM samples, but not for the samples immediately following a
transition. As shown in Appendix~\ref{sec:app-qs}, the resulting error
power, averaged over the slot, decreases as $\Tgamma/\Tr$; this effect is
quantified by the simulations in Section~\ref{sec:val-quasi-static}.

\section{OFDM Signal Analysis}
\label{sec:ofdm-analysis}

Consider a wideband  OFDM waveform with $M$ active subcarriers,
subcarrier spacing $\Delta f$,  useful symbol duration $\Tu = 1/\Delta f$, and
CP duration $\Tg$.
The OFDM signal is  transmitted by a source
located in the far field of the RIS. 
Its complex baseband 
representation is
\begin{equation}
s(t)=\sum_{i\in\Zset}\sum_{m \in \mathcal{M}}s_m[i]\,
e^{j2\pi\frac{m}{\Tu}t}\,\psi(t-i\Ts),
\label{eq:tx}
\end{equation}
where $\psi(t) \eqdef \Ts^{-1/2} \,
\rect\!\left(\frac{t-\Ts/2}{\Ts}\right)$,  
$\Ts\eqdef\Tu+\Tg$ denotes the overall 
OFDM symbol duration,
and $s_m[i]$ is the $i$-th symbol
transmitted over the $m$-th subcarrier.
Subcarriers are indexed symmetrically about the carrier, i.e.,
$m \in \mathcal{M} \eqdef \{-M/2,\dots,M/2-1\}$, with $M$ assumed even, 
so that the $m$-th subcarrier lies at the absolute frequency $\fc + m\Delta f$ 
belonging to $[\fc - B/2,\, \fc + B/2)$, 
with $B = M\Delta f$.

As commonly assumed in RIS-assisted communication systems, 
the direct source--destination link is blocked 
or severely attenuated by
obstacles, so that communication takes place exclusively through
the two-hop RIS-assisted link \cite{Wu2020,DiRenzo2020}.
Moreover, the source, the RIS, the destination,
and the surrounding scatterers are assumed to be static,
or equivalently, the coherence time of the propagation channels 
is assumed to largely exceed the observation interval.
Hence, the source--RIS and RIS--destination links are
modeled as LTI channels, with constant path delays and gains,
and no mobility-induced Doppler shift is considered.
The only source of time variation is the intentional 
periodic reconfiguration of the RIS, which makes the
end-to-end channel LPTV rather than generically time-varying.

Owing to the wideband nature of the transmitted 
OFDM signal, the 
channel between the source and the $q$-th
RIS element is modeled as a
tapped-delay line comprising $L_{\text{s},q}$ 
propagation paths.
Accordingly, the complex baseband signal impinging
on the $q$-th RIS element is 
\begin{equation}
x_q(t)=\sum_{\ell=1}^{L_{\text{s},q}}
h_{\text{s},\ell,q}\, s(t-\tau_{\text{s},\ell,q}),
\label{eq:firsthop}
\end{equation}
where $h_{\text{s},\ell,q}$ and $\tau_{\text{s},\ell,q}$ 
denote, respectively, the complex gain and propagation 
delay of the $\ell$-th path 
from the source to the $q$-th RIS element.

The complex baseband  signal received at the destination,
which is also assumed to lie 
in the far field of the RIS, is the
superposition of the $Q$ individual 
element outputs modified by the
RIS--destination channel:
\begin{equation}
r(t) = \sum_{q=1}^Q \sum_{j=1}^{L_{\text{d},q}}
h_{\text{d},j,q} \, y_q(t-\tau_{\text{d},j,q}) + v(t) \: ,
\label{eq:rx}
\end{equation}
where  $h_{\text{d},j,q}$ and $\tau_{\text{d},j,q}$ 
denote, respectively, the complex gain
and propagation delay of the 
$j$-th path 
from the $q$-th RIS element to the 
destination, and $v(t)$ denotes 
complex additive white Gaussian noise (AWGN) 
at the destination.
Moreover, $y_q(t)$ is related to $x_q(t)$ through 
the LPTV model in~\eqref{eq:bb}.

After CP removal, the receiver  demodulates each  OFDM subcarrier.
For analytical convenience, the receiver timing reference
is aligned with the earliest two-hop propagation path.
Thus, during the $n$-th symbol period, the 
demodulated sample on subcarrier $\bar{h}$  
is%
\footnote{For notational convenience, 
OFDM receiver processing
is represented in continuous time,
although practical receivers perform subcarrier 
demo\-du\-la\-tion digitally by fast Fourier transform (FFT) algorithms.} 
\begin{equation}
y_{\bar h}[n]=\frac{1}{\sqrt{\Ts}}
\int_{n\Ts+\Tg+\tau_{\text{min}}}^{(n+1)\Ts+\tau_{\text{min}}}
r(t)\,e^{-j2\pi\frac{\bar h}{\Tu}t}\,dt,
\label{eq:demod}
\end{equation}
where $\tau_{\text{min}} \eqdef \min_{q,\ell,j} ( \tau_{s,\ell,q} + \tau_{d,j,q} )$.
Substituting \eqref{eq:rx} and \eqref{eq:bb} 
into \eqref{eq:demod},
and using the Fourier transform of \eqref{eq:firsthop},
one obtains the exact input--output relation 
reported in \eqref{eq:twocol} 
at the top of the following page.
\begin{figure*}[!t]
\normalsize
\setcounter{mytempeqncnt}{\value{equation}}
%
\begin{multline}
\label{eq:twocol}
y_{\bar{h}}[n] = \Tu \sum_{q=1}^Q \sum_{j=1}^{L_{\text{d},q}}
\sum_{\ell=1}^{L_{\text{s},q}} h_{\text{d},j,q} h_{\text{s},\ell,q}
\sum_{i\in\Zset} \sum_{m \in \mathcal{M}} s_m[i] \,
e^{j 2 \pi \frac{m}{\Tu} \left(i \Ts + \frac{\Ts}{2}\right)}
\sum_{h \in \Zset}
e^{-j2\pi\left(\frac{\bar{h}}{\Tu}-\frac{h}{T}\right)
\left(n\Ts+\Tg+\frac{\Tu}{2}+\tau_{\text{min}}\right)} \,
e^{-j2\pi \frac{h}{T} \tau_{\text{d},j,q}} \, \eh^\herm \\
\times \int_{-\infty}^{+\infty} \Gammab_q(\nu+\fc) \,
\sinc\left[\Ts\left(\nu-\tfrac{m}{\Tu}\right)\right]
\, \sinc\left[\Tu\left(\tfrac{\bar{h}}{\Tu}-\tfrac{h}{T}-\nu\right)\right]
e^{j2\pi\nu\left[(n-i)\Ts+\frac{\Tg}{2}-(\tau_{\text{s},\ell,q}+\tau_{\text{d},j,q}-\tau_{\text{min}})\right]}
d\nu \, + v_{\bar{h}}[n].
\end{multline}
\setcounter{equation}{\value{mytempeqncnt}}
\hrulefill
\vspace*{4pt}
\end{figure*}
\addtocounter{equation}{1}
Albeit exact, \eqref{eq:twocol} provides limited 
physical insight. 
The following Theorem establishes an equivalent and more
tractable per-subcarrier input--output relation.

\begin{theorem}
\label{thm:guard}
Define 
the two-hop propagation delay spread 
as 
$\Delta\tau_{\text{\rm max}} \eqdef \tau_{\text{\rm max}}-\tau_{\text{\rm min}}$,
where $\tau_{\text{\rm max}} \eqdef \max_{\ell,j,q}  (\tau_{\text{s},\ell,q} + \tau_{\text{d},j,q} )$.
If the CP length obeys
\begin{equation}
\Tg \;\ge\; \Delta\tau_{\text{max}}+\Tgamma,
\label{eq:guard}
\end{equation}
and the RIS reconfiguration period is matched 
to the useful OFDM symbol 
duration, i.e., $T=\Tu$,
then the demodulated samples are free of 
ISI and out-of-grid spectral leakage. 
The $n$-th sample on subcarrier $\bar{h} \in \mathcal{M}$
can be expressed as
\begin{multline}
y_{\bar{h}}[n] = \tfrac{\Tu}{\Ts}
\sum_{m \in \mathcal{M}} \sum_{q=1}^Q
s_m[n] \, \eb_{\bar{h}-m}^\herm \,
\Gammab_q\left(\tfrac{m}{\Tu}+\fc\right) \\
\times \sum_{\ell=1}^{L_{\text{s},q}}
h_{\text{s},\ell,q}
e^{-j 2 \pi \frac{m}{\Tu} \tau_{\text{s},\ell,q}}
\sum_{j=1}^{L_{\text{d},q}}
h_{\text{d},j,q} \,
e^{-j 2 \pi \frac{\bar h}{\Tu} \tau_{\text{d},j,q}}
+ v_{\bar{h}}[n].
\label{eq:io}
\end{multline}
\end{theorem}

\begin{proof}
See Appendix~\ref{sec:app1}.
\end{proof}

\noindent
The result of Theorem~\ref{thm:guard}
is exact for the truncated element response, i.e.,
for the model in which $\gamma_q^{[k]}(\tau)$ is equal to zero for
$\tau \ge T_{\gamma,q}$,
within the quasi-static model \eqref{eq:piecewise};
beyond this model, the structure of \eqref{eq:io} is preserved
and only the coefficients change (see Appendix~\ref{sec:app-qs}).
For the physical response, whose exponential tail has infinite support,
ISI is not identically zero:
it is produced by the portion of the tail falling outside the CP
and is bounded, in amplitude, by the fraction $\chi$
of the initial value that defines $\Tgamma$ in \eqref{eq:Tgamma}.
Accordingly, the residual interference power decreases
at least as fast as $\chi^2$, as verified by the simulations
in Section~\ref{sec:val-guard}.


\subsection{Discussion and Special Cases}

One important consequence of~\eqref{eq:io} 
in Theorem~\ref{thm:guard}
is that each received
subcarrier $\bar h \in \mathcal{M}$ generally contains 
contributions from all transmitted
subcarriers. 
This differs fundamentally from conventional OFDM
transmission over an LTI channel, for which a sufficiently long
CP yields a diagonal per-subcarrier input--output
relation, that is, no ICI.
In the present case, instead,  the periodic RIS reconfiguration translates
each transmitted subcarrier onto other frequencies of the OFDM
grid. 
Consequently, CP removal alone is not sufficient to
diagonalize the input--output relation, even when the CP
satisfies the condition in Theorem~\ref{thm:guard} and completely
suppresses ISI.
The remaining off-diagonal terms 
constitute modulation-induced
ICI, an
intrinsic consequence of the periodic
reconfiguration of the RIS, independent of CP duration and multipath. 

However, unlike generic ICI caused, e.g., by
synchronization errors or Doppler spreading, 
the intercarrier coupling in~\eqref{eq:io} is deterministic 
and highly structured, as discussed in the following Remark.

\begin{remark}[Structured intercarrier coupling]
\label{rem:coupling}
Under the condition $T = \Tu$,
the transmitted subcarrier $m$ 
contributes to the
received subcarrier $\bar{h}=m+h$ through
the $h$-th harmonic of the LPTV RIS response.
From \eqref{eq:io}, 
the corresponding coupling
coefficient is 
\begin{equation}
b_q^{[\bar h-m]}
\left(\frac{m}{\Tu}+\fc\right)
=
\mathbf e_{\bar h-m}^{\herm}
\Gammab_q\left(\frac{m}{\Tu}+\fc\right).
\label{eq:coupling}
\end{equation}
\end{remark}

\noindent
Hence, the RIS generates no out-of-grid spectral components:
all frequency translations occur by integer multiples of the
subcarrier spacing. Energy may nevertheless be transferred to
inactive or out-of-band OFDM-grid frequencies, depending on
the harmonic content of the RIS modulation.

It is instructive to relate Remark~\ref{rem:coupling} to the practice,
recommended in \cite{Bjornson2022} in the context of channel estimation,
of confining RIS state transitions to the silent 
guard intervals between OFDM blocks,
on the grounds that switching is a non-linear operation
that may modulate the reflected signal into adjacent bands.
Theorem~\ref{thm:guard} shows that such a precaution
is not required when reconfiguration is periodic
and matched with the useful symbol duration,
i.e., $T = \Tu$, and the CP length satisfies \eqref{eq:guard}.
Under these conditions, all frequency translations
produced by the surface fall on the OFDM grid
and no out-of-grid leakage is generated,
even though switching occurs
while the useful part of the symbol
is impinging on the surface.

More precisely, the $h$-th harmonic generally 
translates the subcarrier at $m/\Tu$ by
$h/T$ [see \eqref{eq:int-solved} of Appendix~\ref{sec:app1}], 
so that out-of-grid components are absent
if and only if $h/T$ is an integer multiple
of $1/\Tu$, i.e., $h\Tu/T\in\Zset$ for every $h$ with
$b_q^{[h]}(f)\neq0$.
Since the harmonic coefficients are generally nonzero
for every $h$ that is not a multiple of $K$, this requires $\Tu/T=p$, with
$p$ a positive integer, in which case the $h$-th harmonic couples
subcarrier $m$ to $\bar h=m+hp$.

The case $T=\Tu$ is the one of greatest
practical interest. Indeed, $\Tr=T/K=\Tu/(pK)$, so that the switching rate,
which is at most $1/\Tr=pK/\Tu$, is set by three quantities: $\Tu$ is fixed
by the OFDM numerology, and $K$ by the desired harmonic content through the
window in \eqref{eq:fs-inv}, so that $p$ is the only free parameter.
Choosing $p=1$, i.e., $T=\Tu$, yields the lowest switching rate compatible
with on-grid coupling,
which relaxes the requirements on the control circuitry.
Finally, for $T=\Tu$, off-grid leakage arises only when the CP is too
short to absorb the memory of the element; the switching transient,
instead, merely modifies the harmonic coefficients
(see Appendix~\ref{sec:app-qs}).

From the perspective of a conventional OFDM
receiver, harmonic coupling discussed in Remark~\ref{rem:coupling} appears as ICI.
However, since its structure is mainly determined by the RIS control sequence,
it can in principle be mitigated or
exploited. Possible approaches include joint equalization across
coupled subcarriers, transmitter precoding, allocation of
virtual or guard subcarriers, and design of RIS control
sequences that suppress selected harmonics or confine them to
prescribed frequency bins. 
Alternatively, harmonic coupling can
be deliberately exploited for index modulation or waveform
shaping. The development of such transceiver and RIS designs is
beyond the scope of this work.

Besides harmonic coupling, \eqref{eq:io} and \eqref{eq:coupling}
reveal the effect of the frequency-selective RIS response. 
Specifically, the coefficient associated with the coupling from transmitted
subcarrier $m$ to received subcarrier $\bar{h}$ is evaluated at
the absolute input frequency $m/\Tu+\fc$. 
Therefore, different
transmitted subcarriers generally experience different
reflection coefficients, even when they are translated by the
same harmonic order. Harmonic mixing and frequency selectivity
are thus jointly embodied in $b_q^{[\bar h-m]}(m/\Tu+\fc)$ 
and cannot, in general, be
modeled as independent effects.

\begin{remark}[Presence of a direct link]
The presence of a direct source--destination static link 
does not alter the harmonic-coupling 
structure of \eqref{eq:io}. 
Provided that the CP also 
accommodates the delay spread of the
direct link relative to the selected receiver timing reference,
the latter contributes only a conventional diagonal LTI term,
corresponding to $\bar{h}=m$. The off-diagonal terms
$\bar{h}\neq m$ remain entirely determined by the LPTV 
RIS response.
\end{remark}

\noindent
The proposed formulation encompasses conventional
frequency-selective and time-modulated RIS models as special
cases.

\begin{remark}[Frequency-selective time-invariant RIS]
\label{rem:lti-ris}
If the RIS configuration is time invariant (i.e., static), 
only the zeroth harmonic is present in \eqref{eq:fs}, 
namely,
$b_q^{[0]}(f)=\bar\Gamma_q(f)$ and
$b_q^{[h]}(f)=0$, $\forall h \in \Zset -\{0\}$.
The harmonic coupling then disappears (no ICI), and the OFDM
subcarriers remain orthogonal. In this case, \eqref{eq:io}
reduces to
\begin{multline}
y_m[n]
=
\tfrac{\Tu}{\Ts}
s_m[n]
\sum_{q=1}^{Q}
\bar\Gamma_q (\tfrac{m}{\Tu}+\fc )
\\
{}\times
\sum_{\ell=1}^{L_{\text{s},q}}
\sum_{j=1}^{L_{\text{d},q}}
h_{\text{s},\ell,q} \,
h_{\text{d},j,q} \,
e^{-j2\pi\frac{m}{\Tu}
\left(\tau_{\text{s},\ell,q}
+\tau_{\text{d},j,q}\right)}
+
v_m[n].
\label{eq:io-lti}
\end{multline}
\end{remark}

\noindent
Model \eqref{eq:io-lti} is consistent 
with conventional frequency-selective
time-invariant RIS models~\cite{Cai2020,Katsanos2022,Abbas2025,Li2025,Katsanos2024}.
The RIS memory effects remain present in this case 
and must still be
included in the CP condition \eqref{eq:guard}.


\begin{remark}[Frequency-flat time-modulated RIS]
\label{rem:flat-lptv-ris}
If the response associated with each RIS control state is
frequency flat, then
$\Gamma_q^{[k]}(f)=\Gamma_q^{[k]}$
and consequently
$b_q^{[h]}(f)=b_q^{[h]}$ in \eqref{eq:fs}.
The frequency-dependent weighting in \eqref{eq:io} disappears, 
whereas the
harmonic coupling between OFDM subcarriers remains. In
particular, the coefficient coupling subcarrier $m$ to
subcarrier $\bar h$ depends only on the harmonic order
$\bar h-m$, rather than on the absolute input frequency. 
\end{remark}

\noindent
The resulting relation coincides with the conventional
frequency-flat time-modulated RIS model adopted in
space--time-coding and index-modulation
architectures~\cite{Zhang2018,Xiao2024}.
Since a strictly frequency-flat response is instantaneous, the
RIS memory vanishes in this idealized case
($\Tgamma = 0$), and the CP needs to
accommodate only 
the two-hop propagation-channel delay spread.

\begin{remark}[Mutual coupling and beyond-diagonal RIS]
\label{rem:bdris}
If the coupling among RIS elements, either unintentional or through a
beyond-diagonal network, is described by a matrix-valued response
$\boldsymbol{\Theta}(f,t)\in\Cset^{Q\times Q}$ driven by the same periodic,
piecewise-constant control, then \eqref{eq:fs} holds entrywise with
matrix-valued harmonic coefficients $\mathbf{B}^{[h]}(f)$, the harmonic
structure of Remark~\ref{rem:coupling} is unchanged, and
Theorem~\ref{thm:guard} holds with $\Tgamma$ defined from the slowest
decaying pole of the loaded multiport network.
\end{remark}

\noindent
Since the element responses are obtained from a periodic
characterization, the mutual coupling among elements in the same state is
accounted for, within the infinite-array model.
When neighboring elements are in different states,
as required to steer the reflected beam, each element response is
approximated by that of a periodic surface in the same state
(local periodicity approximation), which is the common practice
in RIS design and is accurate for configurations that vary smoothly
across the surface, but degrades for abrupt changes between adjacent elements.

This case, as well as edge effects in finite arrays,
requires a multiport description of the array~\cite{Bidabadi2025}.
Beyond-diagonal RISs require instead a different, multiport meta-atom,
for which frequency-dependent models have been recently
proposed in~\cite{Sena2025,Katsanos2026}.
In both cases, coupling among elements at a fixed frequency and
coupling among subcarriers at a fixed element act on two distinct indices of
the same operator. Since $\mathbf{B}^{[h]}(f)$ depends both on the harmonic
order and on frequency, the operator does not, in general, factor into an
inter-element and an inter-subcarrier term, so that the two effects must be
treated jointly.


\section{Numerical Results}
\label{sec:simulations}

\begin{figure}
\centering
\includegraphics[width=\columnwidth]{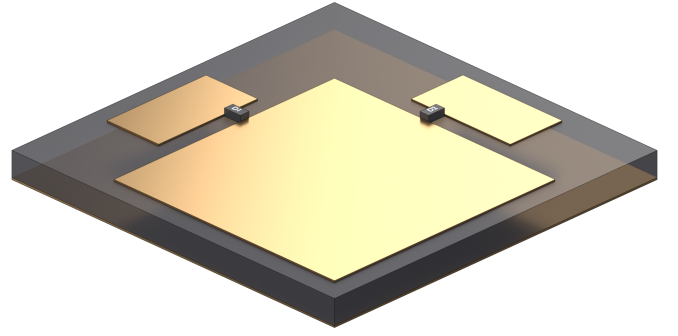}
\caption{Rendered image of the considered simplified unit cell design.
A central square microstrip patch
is coupled to two parasitic patches arranged
along two orthogonal directions.
The gap between the main patch and each parasitic patch
is bridged by a varactor diode (one per linear polarization);
in the EM model, each diode is replaced by a lumped port.
The patches are printed on a grounded F4B substrate,
and the biasing vias and RF chokes of the physical cell
are omitted in this simplified model.}
\label{fig:atom-rendering}
\end{figure}

In this section, we validate the single-pole resonant model
\eqref{eq:polezero} against the full-wave EM response 
of a practical reconfigurable unit cell, and quantify the in-band 
frequency selectivity that motivates the LPTV framework 
of Section~\ref{sec:lptv-model}.
As a case study, we consider the varactor-tuned unit cell of the
open-source OpenRIS design~\cite{Rains2023}, which operates 
in the 5G NR
band n78 ($3.3$--$3.8$~GHz). 
The cell consists of a square microstrip patch 
with two parasitic patches, printed on a grounded F4B substrate
with relative permittivity $\epsilon_r = 2.2$, loss tangent $\tan\delta = 10^{-3}$, 
thickness $d = 2$~mm, and spatial periodicity $30$~mm. Reconfigurability
is provided by two orthogonally oriented
Skyworks SMV1408 varactor diodes~\cite{skyworks_smv1405}.

We adopt the simplified EM model illustrated in 
Fig.~\ref{fig:atom-rendering}.
Specifically, the biasing vias and the RF chokes are omitted,
and the varactor diodes are replaced by lumped ports,
consistently with the modeling choices in~\cite{Rains2023}.
The unit cell is characterized using the open-source
finite-element solver Palace~\cite{awspalace}
(commit \texttt{a0655fb}), built on the MFEM library~\cite{mfem}.

The computational domain comprises a single cell
with doubly-periodic Floquet boundary conditions
on the lateral walls and a Floquet port on the top face,
placed one free-space wavelength at $3.5$~GHz ($\approx 85.7$~mm)
above the patches.
Only the fundamental (specular) mode is retained%
\footnote{Since the period is smaller than the free-space wavelength
over the whole band, all higher-order Floquet modes are evanescent
and are attenuated by more than $140$~dB at the port plane.}.
The metallizations are modeled as zero-thickness perfectly electric conducting sheets,
and each varactor gap as a lumped port with $50$-$\Omega$
reference impedance.
The mesh is adaptively refined using an
a posteriori error estimator over eight iterations,
yielding approximately $1.5\times10^{7}$
unknowns with first-order N\'ed\'elec elements.

The $3$--$4$~GHz response is computed with the adaptive
fast frequency sweep of Palace (relative tolerance $10^{-3}$)
on a $10$-MHz grid, and interpolated onto a $1$-MHz grid
for the subsequent processing.
Reflection phases are referred to the plane
of the patch metallization.
The resulting multiport scattering matrix includes
both the specular Floquet ports and the lumped diode
ports. 
The latter are subsequently terminated
with the varactor equivalent circuit
through the Redheffer star product~\cite{Red1959}.

Each packaged diode is described by the manufacturer SPICE
model~\cite{skyworks_smv1405}: the series inductance
$L_\text{S}$ is connected to the junction branch, i.e.,
the resistance $R_\text{S}$ in series with the junction capacitance
$C_\text{J}(V_\text{R}) =
C_{\text{J}0}\,(1+V_\text{R}/V_\text{J})^{-\alpha}$,
which is shunted by the case capacitance $C_\text{P}$;
the junction diode is neglected, since it is reverse biased.
The model parameters are
$C_{\text{J}0} = 3.89$~pF, $V_\text{J} = 0.92$~V, $\alpha = 0.5$,
$R_\text{S} = 0.6~\Omega$, $C_\text{P} = 0.21$~pF,
and $L_\text{S} = 0.45$~nH for the SOD-882 package~\cite{skyworks_smv1405}.

This procedure allows each applied bias voltage
to be directly mapped onto the corresponding steady-state
control configuration introduced in Section~\ref{sec:lptv-model}.
Owing to the polarization-selective geometry of the unit cell,
the co-polarized reflection coefficient is only weakly affected
by the varactor oriented orthogonally to the incident electric
field. This behavior is an intended feature of the OpenRIS
design~\cite{Rains2023} and is also confirmed by our full-wave
analysis. All results reported below therefore refer to the
co-polarized response obtained when the same bias voltage is
applied to both varactor diodes.

We consider two steady-state control configurations,
denoted by $c_0$ and $c_1$,
corresponding to the reverse-bias voltages
$V_{c_0} = 4$~V and $V_{c_1} = 19.25$~V,
i.e., to $C_\text{J} \approx 1.68$~pF and $0.83$~pF,
respectively.
The two bias states and the channel placement specified below
are  jointly selected so that the cell 
operates as a binary ($1$-bit) phase shifter 
over the full $100$-MHz 5G NR channel.

 
\subsection{Single-Pole Fitting Over the Full Band}
\label{sec:fit}

\begin{figure}
\centering
\includegraphics[width=\columnwidth]{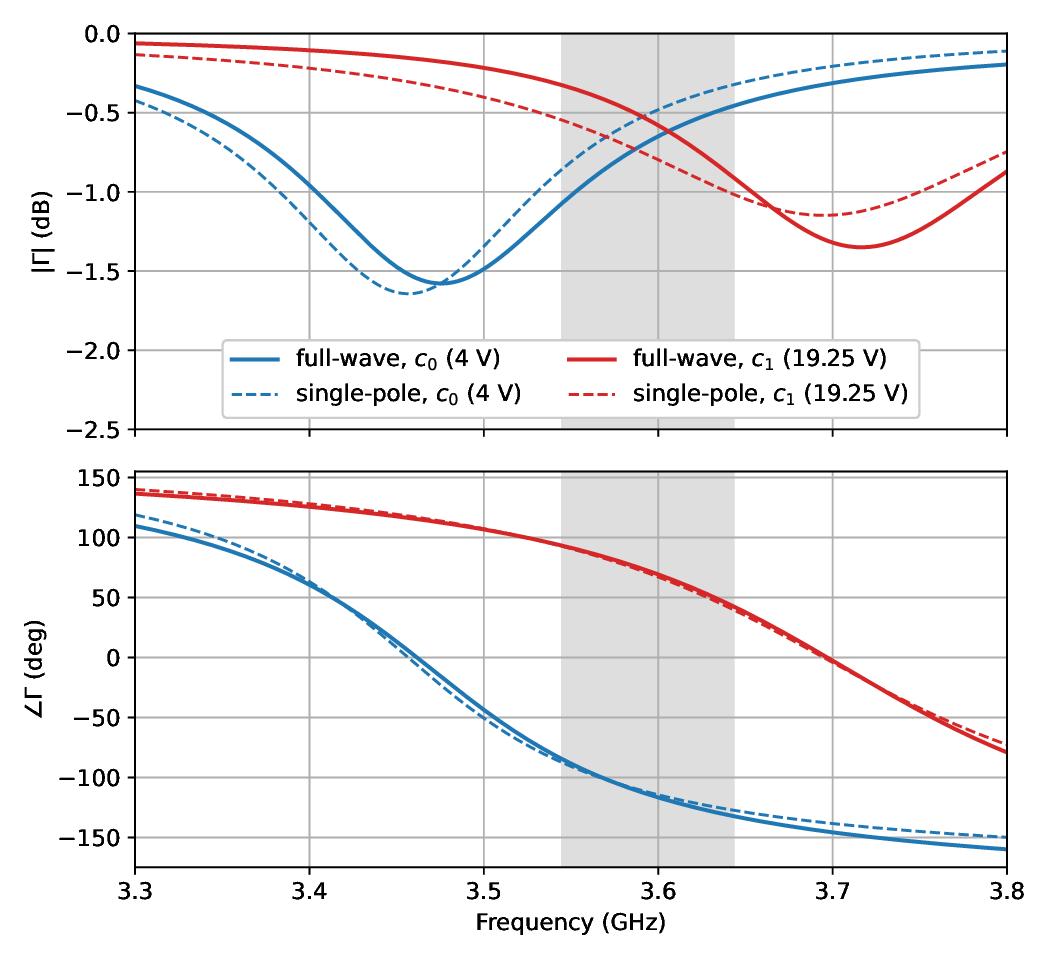}
\caption{Full-wave (solid) versus single-pole (dashed)
bare model \eqref{eq:polezero} for co-polar
reflection of the two control states over the n78 band;
the shaded area marks the selected $100$-MHz channel.}
\label{fig:fit-bare}
\end{figure}

The objective of this experiment 
is twofold: first, to assess the accuracy of the canonical 
single-pole model in~\eqref{eq:polezero}; 
and second, to verify the static phase renormalization predicted
in \eqref{eq:polezero-phi} of Appendix~\ref{sec:surf-response}.
Fig.~\ref{fig:fit-bare} compares the full-wave reflection
coefficient of the two control 
states with the bare single-pole model in~\eqref{eq:polezero}.
Each curve represents the steady-state response of one control state held fixed,
i.e., the frequency-selective LTI description adopted
in~\cite{Cai2020,Abbas2025,Li2025} for that state.
For each state, the three parameters
$(f_{0,q}^{[k]},\xi_{\text{r},q}^{[k]},
\xi_{\text{i},q}^{[k]})$
are estimated by least-squares minimization of 
{the complex-valued error}
over the entire n78 band.
The parameters obtained utilizing the 
renormalization phase $\varphi_0$ model
are reported in Table~\ref{tab:fit},
together with the root-mean-square (RMS) errors of both models.

\begin{table}
\caption{Fitted single-pole parameters of the two control states,
obtained by joint fit of \eqref{eq:polezero-phi} over the n78 band,
with common $\varphi_0 = -10.7^\circ$.}
\label{tab:fit}
\centering
\begin{tabular}{lcc}
\toprule
 & $c_0$ ($4$~V) & $c_1$ ($19.25$~V) \\
\midrule
$f_{0,q}^{[k]}$ (GHz) & $3.471$ & $3.710$ \\
$\xi_{\text{r},q}^{[k]}/2\pi$ (MHz) & $97.5$ & $128.4$ \\
$\xi_{\text{i},q}^{[k]}/2\pi$ (MHz) & $8.9$ & $9.0$ \\
$Q_\text{L}$ & $16.3$ & $13.5$ \\
\midrule
RMS error, bare \eqref{eq:polezero} & $0.104$ & $0.044$ \\
RMS error, with $\varphi_0$ & $0.017$ & $0.023$ \\
\bottomrule
\end{tabular}
\end{table}

Two observations are in order. 
First, the intrinsic decay rate is essentially
independent of the bias voltage, with $\xi_{\text{i},q}^{[k]}/2\pi \approx 8.9$--$9.0$~MHz.
This is consistent with losses dominated by series resistances,
such as that of the varactor, whose contribution
to the intrinsic decay rate does not depend
on the capacitance set by the bias.
By contrast, the bias voltage mainly affects 
the resonance frequency $f_{0,q}^{[k]}$ and the radiative decay 
rate $\xi_{\text{r},q}^{[k]}$. 
In both control states, the unit cell operates in
the strongly overcoupled regime
$\xi_{\text{r},q}^{[k]} \gg \xi_{\text{i},q}^{[k]}$, consistently with
the shallow magnitude notches and low reflection 
clearly visible in Fig.~\ref{fig:fit-bare}.
Second, the RMS error
of the bare model is $0.104$ and $0.044$ for the 
two states, respectively. 
Since the reflection magnitude is close to unity,
these values correspond approximately to relative errors of
$10.4\%$ and $4.4\%$.
Their order of magnitude is consistent with the 
$\mathcal{O}(1/Q_\text{L})$ accuracy floor discussed in
Section~\ref{sec:one-pole} and Appendix~\ref{sec:surf-response},
since the fitted quality factors correspond to $1/Q_\text{L} \approx 6.1$--$7.4\%$.
 
These results confirm that the dominant modeling error
is a static contribution, which can therefore be captured
by the normalized single-pole model \eqref{eq:polezero-phi}.
Since this contribution is nearly common to both control states,
its phase can be estimated through a single joint fit
across the two bias configurations.
\begin{figure}[t]
\centering
\includegraphics[width=\columnwidth]{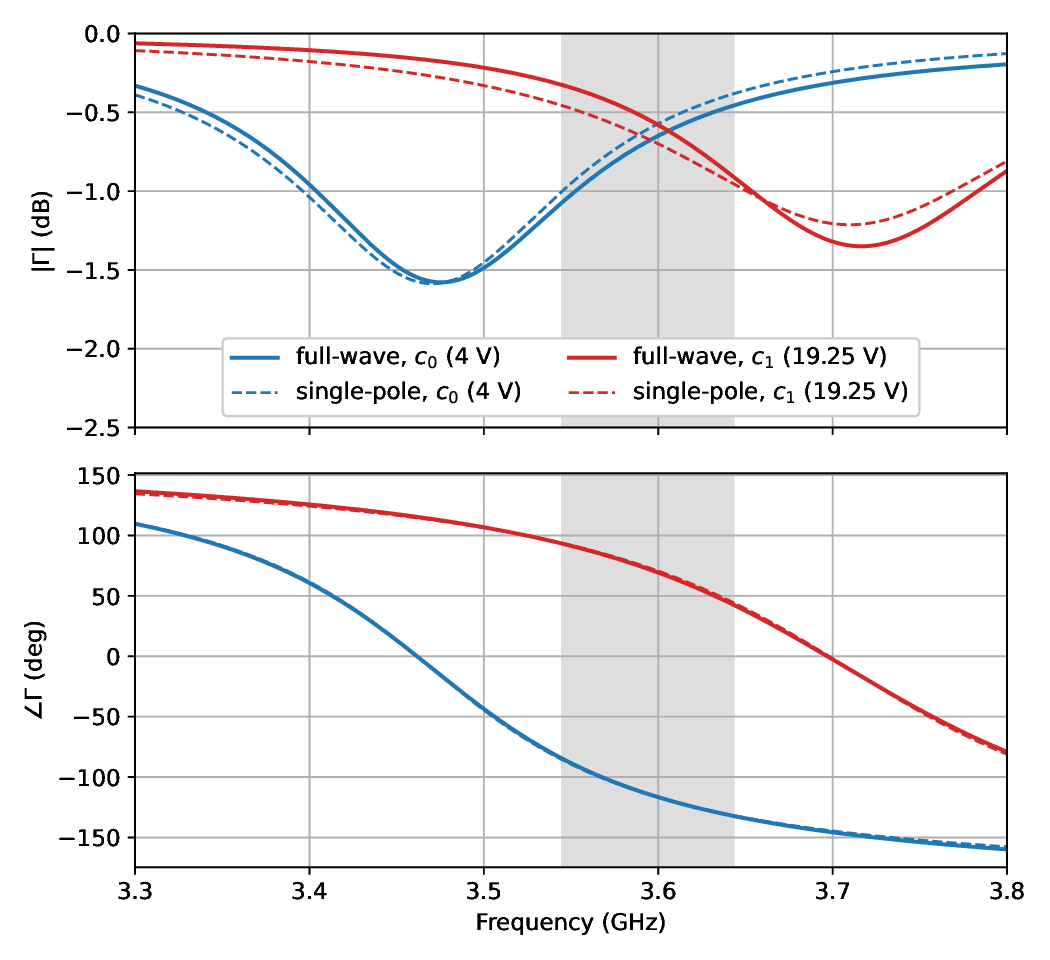}
\caption{Full-wave (solid) versus single-pole (dashed)
offset-aware model \eqref{eq:polezero-phi} for co-polar
reflection of the two control states over the n78 band;
the shaded area marks the selected $100$-MHz channel.}
\label{fig:fit-norm}
\end{figure}
Fig.~\ref{fig:fit-norm} shows the 
result of this joint fit, which
yields the renormalization phase 
$\varphi_0 = -10.7^\circ$ for the 
considered unit cell.
The modeled and full-wave amplitude responses
look quite similar, 
showing an excellent agreement in phase,
while the magnitude presents slight deviations
over the entire $500$-MHz band.
The corresponding RMS errors decrease to $0.017$ and $0.023$
(Table~\ref{tab:fit}), i.e., by factors of about $6$ and $2$,
while the resonant parameters change moderately
(about $15$~MHz, i.e., $0.4\%$, on $f_{0,q}^{[k]}$,
and up to $12\%$ on $\xi_{\text{r},q}^{[k]}$).
As predicted in Appendix~\ref{sec:surf-response}, the common
renormalization factor $e^{j\varphi_0}$ modifies neither the harmonic-coupling
pattern nor the effective memory duration of the RIS element.

 
\subsection{In-Channel Dispersion and Binary Switching}
 
\begin{figure}
\centering
\includegraphics[width=\columnwidth]{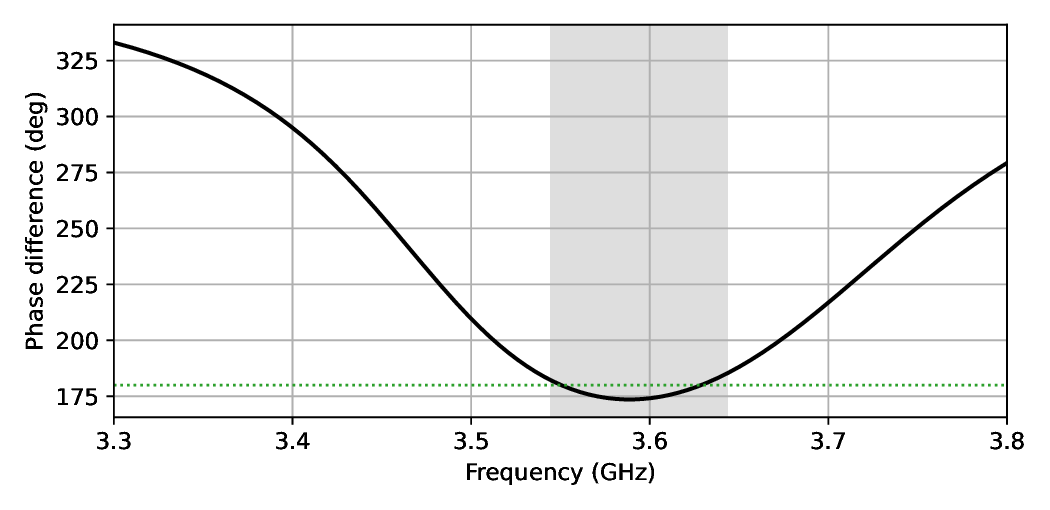}
\caption{Phase difference between the two control states across the n78
band. The shaded area marks the selected $100$-MHz channel, and the green dotted line indicates the ideal 180$^\circ$ response.}
\label{fig:dphi}
\end{figure}
 
We next focus on a nominal $100$-MHz 5G NR channel comprising 
$273$ resource blocks, each containing 12 subcarriers 
with $30$-kHz spacing, for a total 
occupied bandwidth of $98.28$~MHz.
The channel spans the $3.544$--$3.644$~GHz range,
with the occupied bandwidth extending over
$3.5449$--$3.6431$~GHz,
i.e., centered at $\fc$ with the symmetric indexing
of Section~\ref{sec:ofdm-analysis}.
Its spectral position and 
the bias pair $(V_{c_0},V_{c_1})$
are jointly selected so that the two control states
realize a binary phase shift. 
As shown in Fig.~\ref{fig:dphi}, the
phase difference between the two states remains within
$180^\circ \pm 6.5^\circ$ over the entire occupied band.
At the same time, the inter-state magnitude imbalance 
does not exceed $0.75$~dB, and the
reflection loss remains below $1.1$~dB.
 
Despite their nearly ideal differential behavior, the
individual control states are strongly frequency selective.
Over the same $100$-MHz channel, 
their reflection phases vary
by $47^\circ$ and $50^\circ$, respectively, 
while each state
exhibits an in-band magnitude ripple of approximately
$0.6$~dB.
The nearly constant $180^\circ$ phase difference results from 
the two states operating on corresponding flanks of their respective resonances, 
where they exhibit similar phase slopes, or, equivalently, similar group delays.
Consequently, their individual dispersive phase variations
largely cancel when the difference between the two states is taken, although
they remain fully present in each state separately.

The piecewise-constant control
\eqref{eq:piecewise} switches between two control states whose
frequency responses vary appreciably across the OFDM signal band.
Hence, each transmitted subcarrier experiences 
the RIS dispersion at its own absolute frequency, 
as dictated by \eqref{eq:io}.
At the same time, the harmonic-mixing coefficients 
in~\eqref{eq:fs-inv} are determined by the state-dependent 
pole parameters reported in Table~\ref{tab:fit}.


\subsection{Structure of the OFDM Coupling Operator}
\label{sec:operator-structure}

For $T=\Tu$ and $\Tg \ge \Delta\tau_{\text{max}}+\Tgamma$, 
the conditions of Theorem~\ref{thm:guard} are satisfied, 
and \eqref{eq:io} reduces to the linear
per-tone mapping $\mathbf{y}[n]=\Hb\,\mathbf{s}[n]+\mathbf{v}[n]$,
whose entries follow directly from the harmonic coefficients
in ~\eqref{eq:fs-inv} and are given by
\begin{equation}
\left\{\Hb\right\}_{\bar h,m}
= \tfrac{\Tu}{\Ts}\sum_{q=1}^{Q}
A_{\text{d},q}(\bar h)\,
\eb_{\bar h-m}^\herm\,
\Gammab_q\!\left(\tfrac{m}{\Tu}+\fc\right)
A_{\text{s},q}(m),
\label{eq:Hentry}
\end{equation}
where
$A_{\text{s},q}(m)\eqdef\sum_{\ell=1}^{L_{\text{s},q}} h_{\text{s},\ell,q}
e^{-j2\pi\frac{m}{\Tu}\tau_{\text{s},\ell,q}}$ and
$A_{\text{d},q}(\bar h)\eqdef\sum_{j=1}^{L_{\text{d},q}} h_{\text{d},j,q}
e^{-j2\pi\frac{\bar h}{\Tu}\tau_{\text{d},j,q}}$ are 
the frequency responses of the source--RIS and RIS--destination channels associated
with the $q$-th element, evaluated at the transmitted and received subcarrier 
frequencies, respectively.
With $K=2$, the corresponding slot duration is $\Tr=\Tu/2=16.67~\mu$s,
which exceeds $\Tgamma$ by more than three orders of magnitude, so that
the steady-state condition \eqref{eq:settle} of Remark~\ref{rem:steady-state}
is amply satisfied.

This expression
generalizes the conventional OFDM description
of a static RIS-aided link, in which the surface enters through
a single frequency-independent vector multiplying the cascaded propagation response
and the DFT matrix \cite[Eq.~(23)]{Bjornson2022}.
Two generalizations appear here: the surface response
is evaluated at the absolute frequency of each transmitted subcarrier,
and the operator is no longer diagonal,
the off-diagonal entries being indexed by the harmonic order.

Equation~\eqref{eq:Hentry} is the exact 
matrix representation of the per-subcarrier input--output
in~\eqref{eq:io} and therefore follows from it algebraically. 
A direct numerical comparison with a time-domain implementation of the same LPTV
model reproduces~\eqref{eq:Hentry} to machine precision, as
expected.
The physical component requiring validation against an
independent reference is the underlying single-pole unit-cell
response, which was benchmarked against full-wave results in
Section~\ref{sec:fit}.
The purpose of the present subsection is
instead to illustrate the structure of the OFDM coupling
operator predicted by~\eqref{eq:Hentry}. We subsequently
validate the main operational consequence of
Theorem~\ref{thm:guard}, namely the generalized CP
condition, in Section~\ref{sec:val-guard}.

For clarity, the following simulations consider a single RIS
element, i.e., $Q=1$. The results extend directly to $Q>1$:
the overall input--output operator is obtained by the linear
superposition in~\eqref{eq:Hentry}, whereas the effective RIS
memory is determined by the slowest-decaying element according
to~\eqref{eq:Tgamma}. The resulting memory then enters the
CP condition in~\eqref{eq:guard}. The complete set
of simulation parameters is reported in
Table~\ref{tab:sim-params}.

\begin{table}
\caption{Parameters used for 
the numerical validation of Theorem~\ref{thm:guard}.}
\label{tab:sim-params}
\centering
\begin{tabular}{lc}
\toprule
Parameter & Value \\
\midrule
Carrier frequency $\fc$ & $3.594$~GHz \\
5G NR channel & $3.544$--$3.644$~GHz (n78) \\
Channel bandwidth & $100$~MHz \\
Subcarrier spacing $\Delta f$ & $30$~kHz \\
Useful symbol duration $\Tu$ & $33.33~\mu$s \\
Active subcarriers $M$, Fig.~\ref{fig:H-structure} / full & $64$ / $3276$ \\
Active subcarriers $M$, CP sweep & $16$ \\
Control states / control period & $\{c_0,c_1\}$, $K=2$ \\
RIS modulation period $T=\Tu$ / slot $\Tr$ & $33.33~\mu$s, $16.67~\mu$s \\
Cell parameters & see Table~\ref{tab:fit} \\
Memory threshold $\chi$ & $10^{-4}$ \\
RIS memory $\Tgamma$ (worst state, $c_0$) & $13.8$~ns \\
\midrule
Source channel (hop 1) & 3GPP TDL-A~\cite{3gpp38901} \\
Destination channel (hop 2) & 3GPP TDL-C~\cite{3gpp38901} \\
RMS delay spread (per hop) & $30$~ns \\
Combined spread $\Delta\tau_{\text{max}}$ & $549$~ns ($540$ samples) \\
CP-sweep DFT length $N$ & $32768$ \\
Sampling frequency $\fs = N\Delta f$ & $983.04$~MHz \\
Transmitted tone (CP sweep) & $m_0 = 0$ \\
Channel realizations $N_{\text{MC}}$ & $100$ \\
\bottomrule
\end{tabular}
\end{table}

Figure~\ref{fig:H-structure} shows the normalized magnitude
$(\Ts/\Tu)\,|\{\Hb\}_{\bar h,m}|$ for a reduced OFDM system
with $M=64$ active subcarriers,
for a static RIS held in state $c_0$
(Fig.~\ref{fig:H-structure}(a)) and for the
binary control sequence $\{c_0,c_1\}$ corresponding to $K=2$
(Fig.~\ref{fig:H-structure}(b)).
To isolate the contribution 
of the time-modulated dispersive RIS
in the structure of $\Hb$,
both propagation hops are assumed ideal, 
i.e., $h_{\text{s}}(\tau)=h_{\text{d}}(\tau)=\delta(\tau)$,
so that $A_{\text{s},q}(m)=A_{\text{d},q}(\bar h)=1$.
In Fig.~\ref{fig:H-structure},
the main diagonal represents 
the zero-order
reflection response $\eb_0^\herm \, \Gammab_q(m/\Tu+\fc)$, 
whereas  the off-diagonal entries describe 
the time-modulation harmonics of order
$h=\bar h-m$.
Their amplitudes follow the envelope
$|\sinc(h/K)|$ appearing in~\eqref{eq:fs-inv}.
Over the $1.92$~MHz spanned by this reduced system,
the cell response is nearly constant,
so that Fig.~\ref{fig:H-structure} illustrates the harmonic structure
of $\Hb$ rather than its frequency dependence,
which becomes appreciable over the full channel
(see Fig.~\ref{fig:dphi}).

The static case of Fig.~\ref{fig:H-structure}(a)
is the frequency-selective LTI special case \eqref{eq:io-lti}:
only the zeroth harmonic is present, $\Hb$ is diagonal,
and its diagonal is the reflection coefficient of the held state
at each subcarrier frequency, so that no ICI is generated.
Note that exact diagonality still requires the generalized CP
condition \eqref{eq:guard}, since the RIS memory $\Tgamma$ stems from the
dispersion of the element rather than from its reconfiguration.
When the surface is switched (Fig.~\ref{fig:H-structure}(b))
the energy is moved from the diagonal to the odd harmonics:
for the binary sequence considered here the main diagonal
lies about $19$~dB below the first-order harmonics,
since the two held responses are nearly opposite
and their time average almost cancels.

For $K=2$, all nonzero even-order harmonics vanish, 
whereas the odd-order harmonics
decay in magnitude as their order increases.
The resulting approximately banded, on-grid structure, in which 
energy is exchanged between subcarriers at integer harmonic offsets, 
is the frequency-domain 
signature of the time modulation.

\begin{figure*}
\centering
\includegraphics[width=0.85\textwidth]{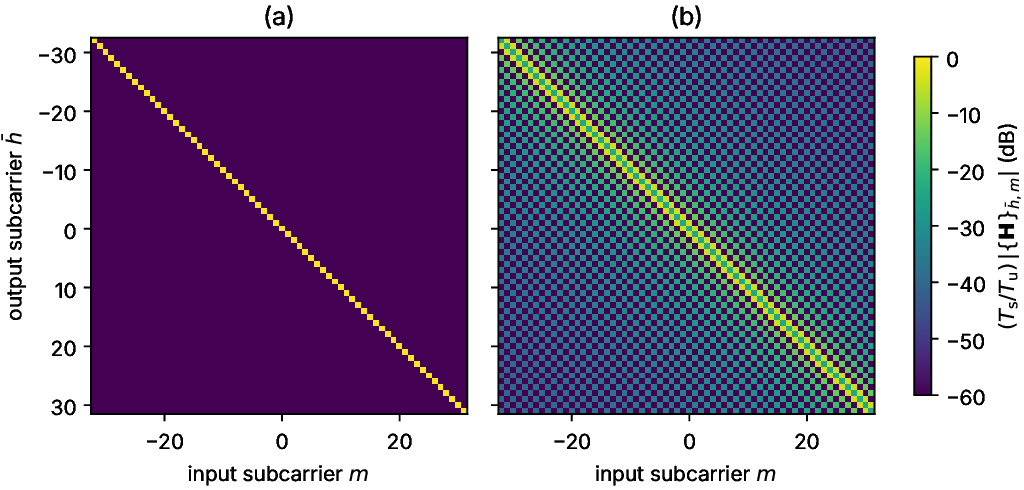}
\caption{Normalized magnitude $(\Ts/\Tu)\,|\{\Hb\}_{\bar h,m}|$
of the OFDM coupling matrix $\Hb$ predicted
by~\eqref{eq:Hentry} for a reduced system with $Q=1$, $M=64$, and
ideal single-tap channels: (a) static RIS held in state $c_0$;
(b) binary control sequence $\{c_0,c_1\}$, $K=2$.}
\label{fig:H-structure}
\end{figure*}


\subsection{Accuracy of the Quasi-Static Approximation}
\label{sec:val-quasi-static}

We evaluate the switching error of Appendix~\ref{sec:app-qs} for the
two-state sequence of Table~\ref{tab:sim-params}, with ideal source and
destination channels. At the carrier,
$\lvert\Gamma^{[0]}(\fc) - \Gamma^{[1]}(\fc)\rvert = 1.85$, since the two
states reflect with nearly opposite phases, and over the occupied band
$R^{[k]} \approx 4$ for both transitions. Immediately after a transition,
the error power therefore exceeds the signal power by about $6$~dB, and the
corresponding samples are not described by \eqref{eq:piecewise}. However,
the transient is short: with $1/\xi^{[k]} = 1.16$--$1.50$~ns, the
instantaneous error falls below the signal within about $1$~ns and below
$-40$~dB within $8$~ns, and $P_e(W)$ in \eqref{eq:qs-pw} equals the signal
power for $W \approx 2.3$--$2.9$~ns. At the sampling frequency of a
$100$-MHz NR carrier, $122.88$~MHz, $\Tgamma = 13.8$~ns spans less than two
samples, so that only one or two samples per transition, out of $4096$ per
symbol, are significantly affected.

Since $T = \Tu$, every OFDM symbol contains up to $K$ transitions.
The normalized error power averaged over the useful interval, i.e.,
the middle term of \eqref{eq:qs-nmse}, is $-38.0$~dB,
while the upper bound in \eqref{eq:qs-nmse}, which depends only on
$\Tgamma/\Tr$ and on the largest jump, gives $-37.5$~dB;
by Parseval's theorem, the latter also bounds the error on the demodulated samples. 
This level is far below the $\mathcal{O}(1/Q_\text{L})$ error of the
fitted single-pole model (Table~\ref{tab:fit}). Since the bound scales as $\Tgamma/\Tr$, the
approximation degrades only when $\Tr$ becomes comparable with
$\Tgamma$.


\subsection{Validation of the Generalized CP Condition}
\label{sec:val-guard}

Finally, we validate the main operational consequence of
Theorem~\ref{thm:guard}, namely the generalized
CP condition \eqref{eq:guard}, according to 
which the CP must accommodate 
both the propagation-channel delay spread 
and the RIS memory.
The two propagation hops are modeled according 
to the tapped delay line profiles 
specified in 3GPP TR~38.901~\cite{3gpp38901}.
Specifically, TDL-A is adopted for the source--RIS link, whereas 
TDL-C (the reference non-line-of-sight wideband profile) is used 
for the RIS--destination link.
Both profiles are configured with an RMS delay spread of $30$~ns.
The resulting end-to-end propagation delay spread
is $\Delta\tau_{\text{max}}=549$~ns, corresponding to 
$540$ simulation samples at the sampling frequency
$\fs=N\Delta f=983.04$~MHz.

To isolate ISI, a reduced system with $M=16$ active subcarriers is
considered. A single unit-amplitude tone on subcarrier $m_0=0$, i.e., at the
carrier, is transmitted in the preceding OFDM symbol, while the target
symbol is silent. After CP removal and demodulation of the target symbol,
the energy collected on the $M$ active subcarriers,
$E_{\text{ISI}}=\sum_{\bar h\in\mathcal{M}}\lvert y_{\bar h}[n]\rvert^2$,
measures the residual ISI. 
It is normalized to the energy $E_{\text{ref}}$ collected
on the same subcarriers when the same tone is transmitted in the target
symbol itself, with the preceding symbol silent and a CP longer than the
generalized guard, so that no ISI is present. The sum over all the active
subcarriers accounts for the fact that the time-modulated RIS spreads the
tone over its harmonics $\bar h=m_0+h$, which are part of the useful
output.

Each hop is normalized to unit total tap power, so that the
reported quantity,
$\eta(\Tg)\eqdef\mathbb{E}\left[E_{\text{ISI}}/E_{\text{ref}}\right]$, is
the residual ISI relative to the useful signal; the ratio is computed for
each channel realization and averaged over $N_{\text{MC}}=100$
independent realizations of the two TDL profiles. 
The tap delays are
deterministic, so $\Delta\tau_{\text{max}}$ is the same in every realization and
only the Rayleigh tap gains are random. The element response is not
truncated: the simulator evaluates the exponential kernel of
\eqref{eq:impulse-response} in closed form, so that $\eta$ is the residual
of the physical response rather than of the truncated model.

Fig.~\ref{fig:cp-guard} reports the resulting normalized residual ISI 
as a function of the CP length for the considered OpenRIS cell, whose
effective memory is $\Tgamma=13.8$~ns, and for
a higher-$Q_\text{L}$ variant with $\Tgamma=68.9$~ns.
The latter is a synthetic configuration,
obtained by scaling both decay rates in Table~\ref{tab:fit}
by a factor of $1/5$, with unchanged resonance frequencies
and propagation channel, rather than from an additional full-wave simulation;
this increases $Q_\text{L}$ to approximately $82$ and $68$.
Two operating regimes can be identified in 
Fig.~\ref{fig:cp-guard}. 
For CP durations shorter than
$\Delta\tau_{\text{max}}$, the ISI decreases rapidly as 
an increasing  fraction of the propagation-channel response is
absorbed by the CP.
Once the CP exceeds $\Delta\tau_{\text{max}}$, 
the remaining interference is mainly
determined by the exponentially 
decaying RIS response and
decreases at a rate governed by the corresponding pole.
The higher-$Q_\text{L}$ response decays five times more slowly and
therefore requires a proportionally longer guard interval. 

In both cases, the generalized bound
$\Delta\tau_{\text{max}}+\Tgamma$, indicated by the dashed lines,
identifies the CP duration at which the resonant tail has
decayed below the design threshold $\chi$. By contrast, the
conventional channel-only bound $\Delta\tau_{\text{max}}$, indicated
by the dotted line, does not account for the residual RIS tail.
The comparison therefore isolates the additional guard
duration associated with the RIS memory.

The residual at the guard point is set by the memory threshold.
Since $\chi$ bounds the amplitude of the slowest-decaying tail,
the residual power at $\Tg=\Delta\tau_{\text{max}}+\Tgamma$ decreases
at least as fast as $\chi^2$, with equality only when the leaking tail
is generated by the slowest state. With the timing alignment
of Section~\ref{sec:ofdm-analysis}, the interference reaching
the target symbol is generated by the RIS during the last slot
of the preceding control period, i.e., in state $c_1$,
whose decay rate exceeds the minimum in \eqref{eq:Tgamma};
the residual is therefore expected to scale as
$\chi^{2\xi_q^{[1]}/\xi_q^{[0]}}=\chi^{2.58}$.
Repeating the measurement for $\chi$ ranging from $10^{-2}$ to $10^{-5}$
confirms this prediction for both cells, once the rounding of the guard
to an integer number of samples is accounted for.
The worst-case definition \eqref{eq:Tgamma} is therefore conservative,
and $\chi$ acts as a design parameter that trades guard duration
against residual interference.

For the considered n78 unit cell, however, 
$\Tgamma=13.8$~ns is much shorter
than both the channel delay
spread and the normal 5G NR CP, whose duration is
approximately $2.3~\mu$s for a $30$-kHz subcarrier spacing.
Thus, in this particular configuration, the CP requirement
remains dominated by the propagation channel, and no practical
extension beyond the standard 5G NR value is required.
The RIS-memory term becomes relevant for higher-$Q_\text{L}$ elements
or shorter OFDM symbols, for which $\Tgamma$ represents a
non-negligible fraction of the available guard interval. This
is the regime addressed by the generalized condition
in~\eqref{eq:guard}.

This conclusion concerns \textit{only} the required
CP duration; it does not justify a memoryless 
model of the RIS element.
Indeed, $\Tgamma$ is a threshold-based effective duration, 
related by \eqref{eq:Tgamma} to the smallest total decay rate
over the surface:
\begin{equation}
\frac{\min_{q,k}\,\xi_q^{[k]}}{2\pi}
=
\frac{\log(1/\chi)}{2\pi\Tgamma}.
\label{eq:memory-bandwidth}
\end{equation}
For $\chi=10^{-4}$ and $\Tgamma=13.8$~ns, the corresponding
characteristic frequency scale is approximately $106$~MHz,
which is comparable to the $98.28$-MHz occupied bandwidth.
Consequently, a memory that is negligible relative to the CP
can still produce significant frequency selectivity across the
OFDM channel band.
In fact, over the selected n78 channel, the reflection phase of
the two control states varies by approximately $47^\circ$ and
$50^\circ$, respectively, whereas their magnitude responses
exhibit ripples of approximately $0.6$~dB. Moreover, the
differential response defining the binary phase-switching
operation varies by $11.5^\circ$ in phase and $1.2$~dB in
magnitude across the channel. A frequency-independent
reflection coefficient is therefore not an adequate model of
the considered unit cell, even though its memory does not
require an extension of the standard CP length.

Furthermore, this frequency selectivity is associated with
control states that change periodically every $\Tr$. The
overall RIS response is therefore LPTV, not only frequency-selective. This joint
behavior is captured by~\eqref{eq:Hentry}: each transmitted
subcarrier is weighted by the dispersive response evaluated at
its own absolute frequency, while the periodic reconfiguration
redistributes energy among the harmonic bands shown in
Fig.~\ref{fig:H-structure}.
Neither effect is removed by increasing the CP duration.
Theorem~\ref{thm:guard} instead guarantees that, once
\eqref{eq:guard} is satisfied, the received signal is free of
ISI and the modulation-induced frequency shifts remain aligned
with the OFDM grid. The resulting structured coupling can then
be addressed through joint subcarrier equalization at the
receiver or through suitable precompensation in the RIS control
design.

\begin{figure}
\centering
\includegraphics[width=\columnwidth]{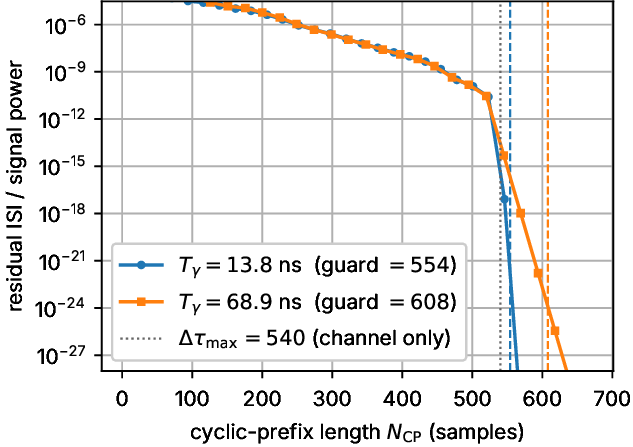}
\caption{Normalized residual ISI versus CP length $N_{\text{CP}}=\Tg\fs$ (in samples)
for the two-hop 3GPP channel, with a TDL-A source--RIS link and a
TDL-C RIS--destination link, each having a $30$-ns RMS delay
spread, over the considered 5G NR n78 channel.
Results are averaged over $100$ channel realizations.}
\label{fig:cp-guard}
\end{figure}


\section{Conclusion}
\label{sec:concl}

This paper developed a physically consistent LPTV model for
wideband OFDM transmission through a dispersive
time-modulated RIS. Unlike conventional models that describe
RIS elements either as frequency-selective LTI systems or as
instantaneous time-varying complex-valued reflection coefficients, 
the proposed
formulation jointly accounts for their frequency selectivity
and periodic reconfiguration.
A resonant equivalent-circuit representation was used to relate
the steady-state reflection coefficients and impulse-response
durations to physically meaningful parameters. 
The resulting harmonic representation preserves the frequency dependence of
each control state while explicitly describing the spectral
components generated by periodic switching.

Based on this model, we derived a closed-form OFDM
input--output relation showing that the received subcarriers are
coupled through the harmonics of the RIS response. This
coupling is deterministic and aligned with the OFDM frequency
grid when the RIS modulation period equals the useful OFDM
symbol duration. 
Frequency selectivity and harmonic mixing
cannot generally be separated, since the coupling between
a transmitted and a received subcarrier depends both
on their harmonic-index difference and on the absolute
frequency of the transmitted subcarrier. 
The familiar frequency-selective time-invariant and frequency-flat
time-modulated RIS models were recovered as special cases.

We also established a generalized CP condition
according to which the CP must accommodate the combined delay
spread of the two-hop propagation channel and the RIS impulse response.
Satisfying this condition eliminates ISI and out-of-grid
spectral leakage, but does not remove the structured
ICI intrinsically generated by RIS reconfiguration. 
Full-wave analysis of an OpenRIS unit cell
confirmed the accuracy of the offset-aware single-pole model
and revealed substantial in-band phase and magnitude variation,
even when the two control states provide an approximately
$180^\circ$ differential phase shift. For the considered 5G NR
configuration, the RIS memory was too short to require an
extension of the standard CP length, whereas simulations with a
higher-$Q_\text{L}$ response confirmed that the additional memory term
becomes relevant as the resonator relaxation time increases.

These results show that an RIS may be negligible as an
additional source of CP overhead while still being strongly
frequency selective over the signal bandwidth. Consequently,
the absence of a CP penalty does not justify modeling the
surface as an instantaneous frequency-flat phase shifter. The
structured nature of the resulting coupling suggests 
developing joint subcarrier equalization, transmitter
precoding, and RIS control strategies that explicitly exploit
the LPTV model. Extensions to multi-resonant elements,
polarization-coupled responses,
mutual coupling among elements in different states,
beyond-diagonal architectures, finite rise-time switching,
and experimental over-the-air validation
constitute relevant directions for future work.


\appendices

\section{Canonical Dispersive Unit Cell}
\label{sec:surf-response}

Throughout this appendix, we consider a generic RIS element
in a generic control state and drop the indices
$q$ and $k$; accordingly, $f_0$, $\xi_{\text{r}}$,
$\xi_{\text{i}}$, $\xi$, and $\Delta$ stand for
$f_{0,q}^{[k]}$, $\xi_{\text{r},q}^{[k]}$,
$\xi_{\text{i},q}^{[k]}$, $\xi_q^{[k]}$,
and $\Delta_q^{[k]}$, respectively.

Following \cite{Costa2021}, 
the three branches of \eqref{eq:Zin} 
can be modeled by the following
impedances:
\begin{align}
Z_d^{\text{pol}}& =jZ_{\text{sub}}^{\text{pol}}\tan(k_zd),
\quad \text{pol} \in \{\text{TE},\text{TM}\}\\
Z_{\text{patch}}^{\text{pol}} & =
R_{\text p}+\frac{1}{j\omega 
C_\text{p}^{\text{pol}}(\theta)}, \quad
\text{pol} \in \{\text{TE},\text{TM}\}\\
Z_{\text{var}} &=R_\text{v}+
j\omega L_\text{v}+\frac{1}{j\omega C_\text{v}} \: .
\label{eq:branches}
\end{align}
where 
$Z_{\text{sub}}^{\text{TE}}=\omega \mu_0/k_z$,
$Z_{\text{sub}}^{\text{TM}}= k_z / (\omega\epsilon_0\epsilon_r)$.
Here, $d$ and $\epsilon_r$ denote the thickness and relative
permittivity of the dielectric substrate, respectively,
while $\epsilon_0$ and $\mu_0$ are the free-space permittivity and
permeability, respectively, and $k_z\eqdef\omega\sqrt{\mu_0\,\epsilon_0}\sqrt{\epsilon_r-\sin^2\theta}$.
In the patch branch, $R_{\text p}$ accounts for the effective
Ohmic losses of the metallic pattern, whereas
$C_{\text p}^{\text{pol}}(\theta)$ is its effective
polarization- and angle-dependent capacitance. The varactor is
represented by a series RLC circuit, where $R_{\text v}$ and
$L_{\text v}$ model its series resistance and package
inductance, respectively, while  $C_{\text v}$ is its tunable
junction capacitance. The latter is controlled through the
applied reverse-bias voltage and determines the corresponding
steady-state control configuration.

All the frequency-domain quantities entering the equivalent
circuit, including the material constitutive parameters and
the lumped-element impedances, are Hermitian symmetric, since
their corresponding time-domain responses are real valued.
Both series and parallel interconnections preserve this
property. Moreover, the reflection-coefficient map
in~\eqref{eq:Gmap} preserves conjugate symmetry because the
free-space wave impedance is real. Consequently, for either
polarization,
\begin{equation}
\Gamma^{\text{pol}}(-f,\theta)
=
\left[
\Gamma^{\text{pol}}(f,\theta)
\right]^*,
\qquad
\text{pol}\in\{\text{TE},\text{TM}\}.
\label{eq:herm}
\end{equation}
Thus, the complete circuit model is conjugate symmetric by
construction. In particular, each steady-state response
$\Gamma_q^{[k]}(f)$ satisfies this property and, since the
switching functions in~\eqref{eq:piecewise} are real valued,
the overall time-varying response obeys
$\Gamma_q(-f,t)=\Gamma_q^*(f,t)$, consistently with the
symmetry used to derive~\eqref{eq:bb}.

Recalling the assumptions introduced in 
Section \ref{sec:one-pole}:
\begin{itemize}

\item[\textbf{(\ref{ass:norm})}]
\textit{Normal incidence and symmetric unit cell:}
$\theta=0$, so that the TE and TM responses
coincide and $Z_0=\zeta_0$, thereby eliminating the 
angular and polarization dependence.

\item[\textbf{(\ref{ass:thin-sub})}]\textit{Electrically thin substrate:}
$k_z d \ll 1$, so that $\tan(k_zd)\approx k_z d$ and
$Z_d \approx j\omega L$, with effective inductance $L=\mu_0 d$.

\item[\textbf{(\ref{ass:cap-load})}]\textit{Capacitively loaded sheet:}
$\omega L_\text{v}\ll1/(\omega C_\text{v})$ and negligible
series reactive loss, so that $Z_{\text{patch}}\|Z_{\text{var}}\approx1/(j\omega C)$,
with total capacitance $C=C_\text{p}+C_\text{v}$.

\item[\textbf{(\ref{ass:loss})}]\textit{Lumped losses:} 
the dielectric, ohmic and varactor losses
are represented by an equivalent shunt conductance $G\ge0$.

\item[\textbf{(\ref{ass:narrow})}]\textit{Narrowband signals:}
$|\omega-\omega_0|\ll\omega_0$, invoked at the linearization step below.
\end{itemize}
Under such assumptions, the equivalent surface input
admittance takes the parallel-resonant Lorentz form
\begin{equation}
Y(\omega)=G+\frac{1}{j\omega L}+j\omega C \: .
\label{eq:Ylor}
\end{equation}
Here, $G$, $L$, and $C$ are the equivalent shunt conductance,
inductance, and capacitance, respectively, of the unit cell. Within the
idealized model, $L=\mu_0 d$ stems from the grounded substrate
[assumption~(\ref{ass:thin-sub})], $C=C_\text{p}+C_\text{v}$ from the patch
and varactor capacitances [assumption~(\ref{ass:cap-load})], and $G\ge0$
lumps the dielectric, ohmic, and varactor losses
[assumption~(\ref{ass:loss})].

For resonant cells whose period is not small
with respect to the wavelength, however, such circuit models are
descriptive rather than predictive, and their lumped elements must be
extracted from the response of the cell~\cite{Costa2014}; they therefore
lose a direct correspondence with the physical components. The physically
meaningful quantities are instead the resonance frequency and the radiative
and intrinsic decay rates introduced below, which are uniquely identified
from the complex reflection coefficient~\cite{Fan2003} and of which $G$,
$L$, and $C$ are merely a reparametrization.
Since $G$, $L$, and $C$ are real valued, \eqref{eq:Ylor}
satisfies $Y(-\omega)=Y^*(\omega)$ and, therefore, preserves the
conjugate symmetry of the complete reflection response
in~\eqref{eq:herm}.

Let us introduce the resonance angular frequency 
$\omega_0 \eqdef 1/\sqrt{LC}$,
under the narrowband assumption (\ref{ass:narrow}), i.e.,
$|\omega-\omega_0|\ll\omega_0$ one obtains
$\omega C-1/(\omega L) \approx 2C(\omega-\omega_0)$
and $Y\approx G+j2C(\omega-\omega_0)$.
Moreover, define the radiative 
(external) and intrinsic (internal) 
rates as
\begin{equation}
\xi_{\text{r}}=\frac{1}{2Z_0C},\qquad
\xi_{\text{i}}=\frac{G}{2C},\qquad
\xi\eqdef\xi_{\text{r}}+\xi_{\text{i}} \: .
\label{eq:rates}
\end{equation}
Substituting in~\eqref{eq:Ylor} and factoring
$\mp2Z_0C$ in numerator and denominator yields
the single-resonance pole--zero reflection coefficient:
\begin{equation}
\Gamma(\omega)
=-\,\frac{j(\omega-\omega_0)+\xi_{\text{i}}-\xi_{\text{r}}}{j(\omega-\omega_0)+\xi_{\text{i}}+\xi_{\text{r}}}
=-1+\frac{2\xi_{\text{r}}}{j(\omega-\omega_0)+\xi}.
\label{eq:pole-zer-app}
\end{equation}
The leading term $-1$ in \eqref{eq:pole-zer-app} 
represents the wave reflection from the grounded substrate,
whereas the second term is the resonant contribution 
associated with a Lorentzian pole of 
half-power half-width $\xi$, centered at $\omega_0$.
Their superposition produces the resonance notch in the magnitude of
the reflection coefficient.

The loaded quality factor and its
external/internal decomposition 
follow from~\eqref{eq:rates} as
\begin{equation}
Q_\text{L}=\frac{\omega_0}{2\xi}=\left(\frac{1}{Q_{\text r}}+\frac{1}{Q_{\text i}}\right)^{-1},
\label{eq:Q}
\end{equation}
where  the radiative and intrinsic quality factors are,
respectively,
\begin{align}
Q_{\text r}=\frac{\omega_0}{2\xi_r}=Z_0\sqrt{\frac{C}{L}} \: , \\
Q_{\text i}=\frac{\omega_0}{2\xi_i}=\frac{1}{G}\sqrt{\frac{C}{L}} \: .
\end{align}
Note that the narrowband assumption is precisely the condition
$\xi\ll\omega_0$, i.e., $Q_{\text L}\gg1$.
 
By retaining only the resonance near $+\omega_0$,
the linearization loses the Hermitian 
symmetry preserved up to~\eqref{eq:Ylor}.
Indeed, the single-pole model~\eqref{eq:polezero} yields
\begin{equation}
\Gamma(-\omega) = -1 + \frac{2\xi_{\text{r}}}{-j(\omega+\omega_0) + \xi} 
\neq\Gamma^{*}(\omega) \: .
\end{equation}
Hermitian symmetry can be restored by including the mirror resonance
centered at $-\omega_0$, which leads to the
Hermitian completion
\begin{equation}
\Gamma_{\text H}(\omega)=-1+\frac{2\xi_{\text{r}}}{j(\omega-\omega_0)+\xi}
+\frac{2\xi_{\text{r}}}{j(\omega+\omega_0)+\xi} \: .
\label{eq:Hcomplete}
\end{equation}
It follows directly that  
$\Gamma_{\text H}(-\omega)=\Gamma_{\text H}^{*}(\omega)$ holds identically.
Equivalently, the inverse transform of~\eqref{eq:Hcomplete} is the real-valued
passband impulse response
\begin{equation}
\gamma_{\text{bp},\text{H}}(t) = -\delta(t)
+4 \, \xi_{\text{r}}\cos(\omega_0t)\,e^{-\xi t}\,u(t).
\label{eq:gH}
\end{equation}
Equation~\eqref{eq:Hcomplete} is a symmetry-restoring approximation,
used here only to exhibit the conjugate symmetry of the passband response
and to bound the error of the one-sided reduction; it is not used to
generate any of the numerical results, which 
use the single-pole model fitted to full-wave data.
In particular, it is not guaranteed to be passive: adding the mirror pole
raises $\lvert\Gamma_{\text H}\rvert$ near resonance by
$\mathcal{O}(1/Q_{\text L})$, so that in the lossless limit
$\xi_{\text{i}}\to0$ it can exceed unity by the same order.
Passivity is instead preserved exactly by the admittance model
\eqref{eq:Ylor}, since $\Gamma=(Y_0-Y)/(Y_0+Y)$ with $Y_0=1/Z_0$ real and
$\Re\{Y\}=G\ge0$ gives $\lvert\Gamma\rvert\le1$ for every
$\omega$; \eqref{eq:Hcomplete} is its narrowband partial-fraction
approximation, and the two differ by $\mathcal{O}(1/Q_{\text L})$ over the
signal band.

It is worthwhile to observe that
the resonance frequency $f_0 \eqdef \frac{\omega_0}{2\pi}$
does not coincide with the carrier frequency. In the binary phase-shifter 
regime considered here, $f_0$ is placed close to, 
but generally not exactly at $\fc$. 
Indeed, at exact resonance the reflection coefficient
collapses to the maximum absorption
$\Gamma(f_0)=-1+2\xi_{\text{r}}/\xi=(\xi_{\text{r}}-\xi_{\text{i}})/(\xi_{\text{r}}+\xi_{\text{i}})$ (zero at critical coupling
$\xi_{\text{r}}=\xi_{\text{i}}$).
The unit cell is instead biased to operate on a 
resonance flank, where the phase slope remains steep while 
the reflection magnitude $|\Gamma|$
remains close to unity. 
In the following analysis, only the scale separation
$|\Delta|\ll \fc$, with $\Delta\eqdef f_0-\fc$, is required; 
the resonance is not assumed to coincide with the carrier.

The real-valued impulse response~\eqref{eq:gH} 
is a {passband} response
characterized by two distinct timescales: 
a rapid oscillation at $f_0$ and an
exponential decay with relaxation time $1/\xi$.
Their separation is determined by
the loaded quality factor~\eqref{eq:Q},
since $f_0/\xi=Q_{\text L}/\pi$.
Since the positive- and negative-frequency 
poles in~\eqref{eq:Hcomplete} have the same decay rate,
the one-pole and two-pole descriptions yield the same effective memory duration
$\Tgamma$. The mirror pole restores the physical carrier
content and Hermitian symmetry of the passband response, but
does not introduce an additional relaxation timescale.

The per-subcarrier relation~\eqref{eq:io} 
samples the surface response only at the
in-band absolute frequencies $f=\nu+\fc$, 
with $\lvert\nu\rvert\le B/2\ll \fc$ and $B = M/\Tu$.
Evaluating the two terms of~\eqref{eq:Hcomplete} there,
the near ($+\omega_0$) term is resonant,
centered at baseband frequency $\nu=\Delta=f_0-\fc$,
while the conjugate ($-\omega_0$) term,
\begin{equation}
\frac{2\xi_r}{j2\pi(\nu+\fc+f_0)+\xi}\;\approx\;\frac{2\xi_r}{j2\pi(\fc+f_0)},
\qquad|\nu|\ll \fc,
\label{eq:conj}
\end{equation}
is slowly varying and bounded, relative to the prompt term, by
\begin{equation}
\frac{2\xi_r}{2\pi(\fc+f_0)}
\lesssim\frac{\xi}{\pi(\fc+f_0)}\approx\frac{1}{2Q_{\text L}},
\label{eq:bound}
\end{equation}
where the last step uses $\fc+f_0\approx2f_0$ and $\xi=\pi f_0/Q_{\text L}$.
Moreover, the conjugate pole sits at $\nu\approx-(\fc+f_0)$, far below the band: within
$\lvert\nu\rvert\le B/2$ the term~\eqref{eq:conj}
is on the smooth tail of that 
distant pole and
varies only by a relative $\mathcal{O}(B/\fc)$. Indeed, one has
\begin{multline}
\left|\frac{g_-(\nu)-g_-(0)}{g_-(0)}\right|
=\frac{2\pi\lvert\nu\rvert}{\big|\,j2\pi(\nu+\fc+f_0)+\xi\,\big|}\\
\le\frac{\lvert\nu\rvert}{\fc+f_0-\lvert\nu\rvert}
\le\frac{B/2}{\fc+f_0-B/2}\approx\frac{B}{2(\fc+f_0)},
\label{eq:flat}
\end{multline}
where $g_-(\nu)$ denotes the left-hand side of~\eqref{eq:conj}
and the overall rate $\xi$.
Hence, the term is essentially constant across the band:
it merely renormalizes the prompt reflection, i.e., 
the $\delta(t)$ term in~\eqref{eq:gH},
and adds neither resonance nor memory.
The in-band selectivity and the finite memory~\eqref{eq:Tgamma}
are therefore set by the near ($+\omega_0$) 
pole alone, whose resonance lies 
at $\nu=\Delta$, close to the band,
and the one-sided (single-pole) model represents the
cell with relative error $\mathcal{O}(1/Q_{\text L})$.
It may thus be used for $\Gamma^{[k]}_q(f)$
in~\eqref{eq:fs-inv}--\eqref{eq:io}.

The above analysis shows that the terms neglected by the single-pole reduction,
namely the conjugate-pole contribution and the residual background
of the grounded slab beyond the thin-substrate approximation,
are not dispersive over the signal band. Their variation is only
$\mathcal{O}(B/\fc)$, so that, to the order of the model,
they act as a static multiplicative factor.
Accordingly, the aggregate contribution of these terms
can be absorbed into a constant renormalization 
phase factor that multiplies the reduced response,
\begin{equation}
\Gamma_q^{[k]}(f) \approx e^{j\varphi_0}\,
\Gamma_{\text{pole},q}^{[k]}(f),
\label{eq:polezero-phi}
\end{equation}
where $\Gamma_{\text{pole},q}^{[k]}(f)$ is given by \eqref{eq:polezero}.
The phase offset $\varphi_0$ is common to the control states
only within $\mathcal{O}(1/Q_\text{L})$;
the state dependence enters through the conjugate-pole term
and is of that order.
Being common to all control responses, $e^{j\varphi_0}$
simply scales the Fourier coefficients
in \eqref{eq:fs}--\eqref{eq:fs-inv}
without affecting memory \eqref{eq:Tgamma}
or the guard-interval condition of Theorem~\ref{thm:guard}.

\section{Accuracy of the Quasi-Static Approximation}
\label{sec:app-qs}

We consider a single element and drop for convenience 
the index $q$. Let $x(t)$ and $y(t)$
be the complex envelopes of the incident and reflected signals, and let
$p^{[k]} \eqdef -\xi^{[k]} + j2\pi\Delta^{[k]}$ be the pole of the $k$-th
state. Since, by \eqref{eq:polezero},
$\Gamma^{[k]}(\nu+\fc) = -1 + 2\xi_\text{r}^{[k]}/(j2\pi\nu - p^{[k]})$,
the input--output relation of the element in the $k$-th state can be
equivalently described by 
the differential equation
\begin{equation}
\dot{y}(t) - p^{[k]} y(t) = -\dot{x}(t)
+ \left(p^{[k]} + 2\xi_\text{r}^{[k]}\right) x(t) \: .
\label{eq:qs-ode}
\end{equation}
Let a control slot start at $\ts$, with state $k$ held in
$[\ts, \ts+\Tr)$ and state $k-1$ (modulo $K$) held in the preceding slot.
Solving \eqref{eq:qs-ode} gives, for $t \in [\ts, \ts+\Tr)$,
\begin{multline}
y(t) = -x(t) + \left[y(\ts) + x(\ts)\right] e^{p^{[k]}(t-\ts)} \\
+ 2\xi_\text{r}^{[k]} \int_{\ts}^{t} e^{p^{[k]}(t-\alpha)}\, x(\alpha)\, d\alpha \: .
\label{eq:qs-sol}
\end{multline}
On the other hand, the quasi-static output prescribed by \eqref{eq:piecewise}
is the inverse Fourier transform of $\Gamma^{[k]}(\nu+\fc)\,X(\nu)$, i.e.,
\begin{equation}
y_\text{qs}(t) = -x(t) + 2\xi_\text{r}^{[k]} \int_{-\infty}^{t}
e^{p^{[k]}(t-\alpha)}\, x(\alpha)\, d\alpha \: .
\label{eq:qs-ss}
\end{equation}
Splitting the integral in \eqref{eq:qs-ss} at $\ts$, the part over
$[\ts, t]$ cancels the corresponding term of \eqref{eq:qs-sol}, while the
part over $(-\infty, \ts)$ equals
$[y_\text{qs}(\ts^+) + x(\ts)]\, e^{p^{[k]}(t-\ts)}$.

We now assume that
(i) the output is continuous at the switching instant, i.e.,
$y(\ts^-)=y(\ts^+)$, which amounts to the continuity of the RF voltage
across the sheet, whose envelope is $x(t)+y(t)$ in the equivalent circuit
of Appendix~\ref{sec:surf-response}, since the bias transition spans many
carrier periods, and (ii) the preceding slot has lasted long enough
for the element to reach the steady state of state $k-1$, i.e.,
$y(\ts)=y_\text{qs}(\ts^-)$.\footnote{The last holds up to a residual
of order $\chi$ for $\Tr \ge \Tgamma$.}
Then the error $e(t) \eqdef y(t) - y_\text{qs}(t)$ can be written as
\begin{equation}
e(t) = \deltas\, e^{p^{[k]}(t-\ts)}, \quad t \in [\ts, \ts+\Tr) \: ,
\label{eq:qs-error}
\end{equation}
where $\deltas \eqdef y_\text{qs}(\ts^-) - y_\text{qs}(\ts^+)$ represents
the jump of the quasi-static output at $\ts$.

The error power averaged over a window of duration $W \le \Tr$ starting at
$\ts$ can be evaluated as
\begin{equation}
P_e(W) \eqdef \frac{1}{W} \int_{\ts}^{\ts+W} \lvert e(t) \rvert^2 dt
= \lvert\deltas\rvert^2\, \frac{1 - e^{-2\xi^{[k]} W}}{2\xi^{[k]} W} \: .
\label{eq:qs-pw}
\end{equation}
Hence, $P_e(W) \approx \lvert\deltas\rvert^2$ for $W \ll 1/(2\xi^{[k]})$,
i.e., the samples immediately following $\ts$ are affected by an error as
large as the jump, whereas $P_e(W) \approx \lvert\deltas\rvert^2 /
(2\xi^{[k]} W)$ for $W \gg 1/(2\xi^{[k]})$.

Over a useful interval of duration $\Tu = T$, the element undergoes
at most $K$ transitions, since $\deltas^{[k]} = 0$ whenever
$\Gamma^{[k-1]} = \Gamma^{[k]}$. Since $\Tr \gg \Tgamma$, each transient is
exhausted within its slot and, by \eqref{eq:qs-pw} for $W \to \infty$, carries the
energy $\lvert\deltas^{[k]}\rvert^2/(2\xi^{[k]})$, where $\deltas^{[k]}$ is
the jump of the quasi-static output at the $k$-th transition. Hence, the
error power averaged over the interval is
\begin{equation}
\bar{P}_e \approx \frac{1}{\Tu} \sum_{k=0}^{K-1}
\frac{\lvert\deltas^{[k]}\rvert^2}{2\xi^{[k]}} \: .
\label{eq:qs-pbar}
\end{equation}
As a function of $\ts$, $\deltas^{[k]}$ is the output of the LTI filter
with frequency response $\Delta\Gamma^{[k]}(\nu) \eqdef
\Gamma^{[k-1]}(\nu+\fc) - \Gamma^{[k]}(\nu+\fc)$ driven by $x(t)$. Its power
is therefore $\int_{-\infty}^{+\infty}
\lvert\Delta\Gamma^{[k]}(\nu)\rvert^2 S_x(\nu)\, d\nu$, where $S_x(\nu)$ is
the power spectral density of $x(t)$, whereas the power of the
quasi-static output in the $k$-th slot is $P_y^{[k]} =
\int_{-\infty}^{+\infty} \lvert\Gamma^{[k]}(\nu+\fc)\rvert^2 S_x(\nu)\,
d\nu$. Denoting by $R^{[k]}$ their ratio and using
$1/\xi^{[k]} \le \Tgamma/\log(1/\chi)$, the normalized error power is
\begin{equation}
\frac{\mathbb{E}[\bar{P}_e]}{P_y} \approx \sum_{k=0}^{K-1}
\frac{R^{[k]}}{2\xi^{[k]}\Tu}
\le \frac{\Tgamma}{2\log(1/\chi)\,\Tr}\, \max_k R^{[k]} \: ,
\label{eq:qs-nmse}
\end{equation}
where $P_y^{[k]} \approx P_y$ has been assumed for all $k$
and $K/\Tu = K/T = 1/\Tr$ has been used.
Since the error propagates linearly through the RIS--destination channel,
by Parseval's theorem \eqref{eq:qs-nmse} also bounds, up to the gain of
that channel, the error power on the demodulated samples
\eqref{eq:demod}.

Finally, the quasi-static approximation affects only
the values of the harmonic coefficients,
not the structure of Theorem~\ref{thm:guard}.
Indeed, the switched element described by \eqref{eq:qs-ode},
with $p^{[k]}$ and $\xi_\text{r}^{[k]}$ periodically switched,
is itself an LPTV system of period $T$,
whose impulse response decays at least as fast as
$e^{-\min_k \xi^{[k]}\tau}$ and is therefore negligible
beyond $\Tgamma$.
Hence, the proof in Appendix~\ref{sec:app1}
applies unchanged once $\eh^\herm\gammab_q(\tau)$
is replaced by the corresponding Fourier coefficient
of the exact kernel, whose effective support still is $[0,\Tgamma)$:
for $T=\Tu$ and under \eqref{eq:guard},
the switching transient generates neither ISI nor out-of-grid components,
and only perturbs $b_q^{[h]}(f)$, with an effect
on the demodulated samples bounded by \eqref{eq:qs-nmse}.

\section{Proof of Theorem~\ref{thm:guard}}
\label{sec:app1}

The inner integral factor in \eqref{eq:twocol}
represents in the time domain a cascade
of two rectangular impulse responses
(the OFDM symbol window $\Ts$ and the discrete Fourier transform (DFT) window $\Tu$)
convolved with the surface impulse response
$\gammab_q(t)$.
Let us introduce
\begin{equation}
\boldsymbol{z}(t)
\eqdef
{\gammab}_q(t) * g(t)
= \int_{0}^{\Tgamma}
{\gammab}_q(\tau)
\, g(t-\tau) \, d\tau,
\label{eq:inner-int}
\end{equation}
where ${\gammab}_q(t) \eqdef [
{\gamma}_q^{[0]}(t),{\gamma}_q^{[1]}(t),
\ldots,{\gamma}_q^{[K-1]}(t)]^\trasp \in \Cset^K$
and
\begin{equation}
g(t) \eqdef \tfrac{1}{\Ts\Tu} 
\left[\rect(\tfrac{t}{\Ts})
e^{j 2 \pi \frac{m}{\Tu} t}\right] *
\left[\rect(\tfrac{t}{\Tu}) \, 
e^{j 2 \pi (\frac{\bar{h}}{\Tu} - \frac{h}{T}) t}\right] \: .
\end{equation}
It can be readily shown that
\begin{equation}
g(t) \equiv 0, \quad \forall t \not\in
\left[-\Tu-\tfrac{\Tg}{2}, \Tu+\tfrac{\Tg}{2}\right].
\end{equation}
Furthermore,
\begin{equation}
g(t) = \tfrac{1}{\Ts} \, 
\sinc\left[ \Tu \left(
\tfrac{\bar h}{\Tu} - \tfrac{h}{T} - \tfrac{m}{\Tu}
\right)\right] e^{j2 \pi \frac{m}{\Tu} t},
\label{eq:g-central}
\end{equation}
for all $|t| \leq \Tg/2$.
For $n=i$, from \eqref{eq:inner-int} 
and the inner-integral term in
\eqref{eq:twocol}, we require
\begin{equation}
-\tfrac{\Tg}{2} \leq \tfrac{\Tg}{2}
- (\tau_{\text{s},\ell,q}+\tau_{\text{d},j,q}-\tau_{\text{min}}+\tau)
\leq \tfrac{\Tg}{2},
\end{equation}
where $\tau \in [0,\Tgamma)$.
The right-hand condition is satisfied by definition.
The left-hand condition becomes
\begin{equation}
\Tg \geq \tau_{\text{s},\ell,q}+\tau_{\text{d},j,q}-\tau_{\text{min}}+\tau,
\end{equation}
whose most stringent form is
\begin{equation}
\Tg \geq \Delta\tau_{\text{max}} + \Tgamma.
\end{equation}
Under this condition, it follows that all contributions corresponding to
$n\neq i$ vanish and%
\begin{multline}
\boldsymbol{z}(t) =
\tfrac{1}{\Ts} \, \sinc\left[ \Tu \left(
\tfrac{\bar h}{\Tu} - \tfrac{h}{T} - \tfrac{m}{\Tu}
\right)\right]  \\
\times 
\Gammab_q\left(\tfrac{m}{\Tu}+\fc\right)
e^{j2 \pi \frac{m}{\Tu} t} \: ,
\label{eq:int-solved}
\end{multline}
for all $\Tgamma-\Tg/2 \leq t \leq \Tg/2$.
Indeed, \eqref{eq:inner-int} integrates $g(t-\tau)$
over $\tau\in[0,\Tgamma]$, so the closed form of $g(\cdot)$
is applicable only if $\lvert t-\tau\rvert\le\Tg/2$
for every such $\tau$, which yields the stated interval rather than
$\lvert t\rvert\le\Tg/2$.
The times at which \eqref{eq:int-solved}
is used in \eqref{eq:twocol} for $n=i$, namely
$t=\Tg/2-(\tau_{\text{s},\ell,q}+\tau_{\text{d},j,q}-\tau_{\text{min}})$,
belong to $[\Tg/2-\Delta\tau_{\text{max}},\,\Tg/2]$
and therefore satisfy $t\ge\Tgamma-\Tg/2$ if and only if
$\Tg\ge\Delta\tau_{\text{max}}+\Tgamma$, i.e.,
precisely under the CP condition \eqref{eq:guard}
already imposed above.
Hence \eqref{eq:int-solved} holds at every time at which it is invoked.
If $T=\Tu$, one has that
\begin{equation}
\sinc\left[ \Tu \left(
\tfrac{\bar h}{\Tu} - \tfrac{h}{\Tu} - \tfrac{m}{\Tu}
\right)\right] = \delta_{h-(\bar{h}-m)} \: .
\label{eq:kron}
\end{equation}
Substituting \eqref{eq:int-solved} and \eqref{eq:kron}
into \eqref{eq:twocol}, one obtains \eqref{eq:io}.

\IEEEtriggeratref{21}
\bibliographystyle{IEEEtran}
\bibliography{biblio}

\end{document}